\documentclass{article}
\usepackage[margin=1in]{geometry}

\usepackage{amssymb}
\usepackage{amsmath}
\usepackage{amsthm}
\usepackage{bbm}
\usepackage{mathtools}
\usepackage{enumitem}
\usepackage{todonotes}
\usepackage{algorithm}
\usepackage{algpseudocode}
\usepackage{comment}
\usepackage{titling}
\usepackage{multicol}
\usepackage{slashed}
\usepackage{tikz}
\usepackage{adjustbox} 
\usepackage{tikz-3dplot}
\usetikzlibrary{backgrounds}
\usetikzlibrary{plotmarks}
\usetikzlibrary{shapes, calc}
\usepackage{latexsym}

\usepackage{marginnote}
\usepackage{caption}
\usepackage{hyperref}
\usepackage{cleveref}
\usepackage{thmtools}

\newcommand{\randomwalks}{\textsc{RandomWalks}}
\newcommand{\approxdist}{\textsc{SpectralDistance}}
\newcommand{\SDP}{\textsc{SpectralDotProduct}}
\newcommand{\Color}{\textsc{Color}}
\newcommand{\prepro}{\textsc{FindRepresentatives}}

\newtheorem{thm}{Theorem}[section]
\newtheorem{definition}[thm]{Definition}
\newtheorem{lemma}[thm]{Lemma}
\newtheorem{claim}[thm]{Claim}

\newtheorem{remark}[thm]{Remark}
\newtheorem{cor}[thm]{Corollary}

\newtheorem{obs}[thm]{Observation}
\newtheorem*{term*}{Terminology}
\newtheorem*{model*}{Input model}
\newtheorem*{goal*}{Algorithmic goal (AG)}
\newtheorem*{sett*}{Parameter setting (PS)}
\newtheorem*{question*}{Question (Q)}
\newtheorem*{lemma*}{Lemma}

\newcounter{todocounter}

\DeclareMathOperator*{\argmin}{arg\,min}
\DeclareMathOperator*{\median}{median}

\newcommand{\R}{\ensuremath{\mathbb{R}}}
\DeclareMathOperator*{\E}{\ensuremath{\mathbb{E}}}
\DeclareMathOperator*{\Var}{\ensuremath{Var}}

\newcommand{\1}{\ensuremath{\mathbbm{1}}}
\newcommand{\U}{\ensuremath{U_{[k-1]}}}
\newcommand{\Ustar}{\ensuremath{U_*}}
\newcommand{\Astar}{\ensuremath{\bar A_{\star}}}

\newcommand{\diag}{\ensuremath{\text{diag}}}
\newcommand{\spn}{\ensuremath{\mathrm{span}}}
\newcommand{\supp}{\ensuremath{\mathrm{supp}}}

\newcommand{\nei}{\ensuremath{\mathrm{N}}}
\newcommand{\coeff}{\ensuremath{\mathrm{coeff}}}

\newcommand{\apx}{\ensuremath{\mathrm{apx}}}

\title{Recovering Planted Colorings in Sublinear Time}
\author{Weronika Wrzos-Kaminska \\ EPFL}
\date{}
\begin{document}
\begin{titlingpage}

\maketitle

\begin{abstract}
\noindent 
We give a sublinear algorithm for the planted $k$-coloring problem. Given an expander $G$ with a planted coloring, the goal is to efficiently determine the color class of a given vertex. 

We work in the adversarial planted coloring model of David and Feige [STOC 2016], where an adversary chooses a $d$-regular spectral $\lambda$-expander $G$ on $n$ vertices and plants a balanced $k$-coloring by partitioning the vertices into $k$ equal parts and deleting all edges within each part. This model generalizes the earlier random graph models studied by  Blum and Spencer [J. Algorithms 1995] and  Alon and Kahale [STOC 1994]. 

We give the first sublinear-time algorithm for recovering planted colorings in this model.
The algorithm has preprocessing time and space $\widetilde O\left(n^{1/2+O(1/\log(d/\lambda))}\right)$, and produces a data structure that answers color queries in time $\widetilde O\left(n^{1/2+O(1/\log(d/\lambda))}\right)$, such that the resulting labeling agrees with the planted coloring on all but an $O(\sqrt{\lambda /d})$ fraction of vertices, up to a permutation of the $k$ colors. 

The algorithm gives sublinear-time inner product access to the bottom eigenspace of the normalized adjacency matrix, which allows us to adapt the classical spectral approach of Alon and Kahale in sublinear time.
\end{abstract}

\end{titlingpage}
\setcounter{page}{0}
\setcounter{tocdepth}{2}
\tableofcontents

\newpage

\section{Introduction}
Graph coloring is a fundamental problem in theoretical computer science. In the \(k\)-coloring problem, the goal is to assign each vertex a color from \([k]\) so that adjacent vertices receive different colors.
Finding $k$-colorings in graphs for $k \geq 3$ is among the first computational problems shown to be NP-hard \cite{Karp72, NPcomplete}. This has motivated extensive work on approximation algorithms and on identifying assumptions under which the problem can be solved efficiently \cite{Wigderson83, BR90, Blum94,BlumSpencer, BK97, AK97, KMS98, KLS00, Khot01, AC06, DMR06, AG11,  DF16, KT17, Kumar17, KTY24, Buhai25, Hsieh26}. 

One influential line of work studies planted coloring models, where a hidden coloring is embedded into a graph, and the task is to recover (approximately) the planted solution. Blum and Spencer \cite{BlumSpencer} and Alon and Kahale \cite{AK97} considered random graph models with a planted $k$-coloring, and showed that polynomial-time spectral algorithms can recover the planted $k$-coloring with high probability. David and Feige \cite{DF16} extended this to a more challenging adversarial setting, in which an adversary chooses an expander and plants a balanced 
$k$-coloring by deleting all edges within each color class. They showed that, even in this adversarial setting, the spectral approach of Alon and Kahale recovers the planted coloring on a $0.99$ fraction of the vertices, while exact recovery is NP-hard.

While planted coloring models have been extensively studied in the polynomial-time regime, the question of \emph{sublinear-time} recovery remains largely unexplored. This motivates the central question of this work:
\vspace{0.1in}

\fbox{
\parbox{0.9\textwidth}{
\begin{center}
Is it possible to recover a planted $k$-coloring in sublinear time?
\end{center}
}
}

\vspace{0.1in}

Of course, any algorithm that explicitly outputs a coloring requires $\Omega(n)$ time simply to write down the labels. We therefore focus on constructing a \emph{coloring oracle}: a small-space randomized data structure that quickly answers queries of the form “what color does vertex $v$ have?”. The oracle must be \emph{consistent}, meaning that its answers agree with a single underlying assignment of colors (not necessarily a proper $k$-coloring), and \emph{approximately correct}, meaning that the oracle disagrees with the planted coloring on only a small fraction of vertices, up to a permutation of colors (see \Cref{def:oracle} for a formal definition of coloring oracle). 

We consider a variant of the adversarial planted model of $\cite{DF16}$. In this model, an adversary chooses a $d$-regular spectral $\lambda$-expander (i.e., eigenvalues of its adjacency matrix, other than the top eigenvalue \(d\), have absolute value at most \(\lambda\)). The adversary also chooses a balanced partitioning of the vertices into $k$ color classes, for a constant $k$.  The planted instance is obtained by deleting all the monochromatic edges, i.e., all edges that have both endpoints in the same color class. Additionally, we require that the partitioning is chosen so that the degrees in the planted instance are balanced, in the sense that all vertices have degree $\Omega(d)$ after deleting the monochromatic edges (see \Cref{def:model} for a formal definition of the input model).  

A $d$-regular graph has $\lambda_1 = d$ and that the second largest eigenvalue in absolute value satisfies $\lambda = \Omega(\sqrt{d})$ \cite{Nilli}.
Random \(d\)-regular graphs are essentially best possible spectral expanders, satisfying
\(\lambda = O(\sqrt{d})\) with high probability \cite{FO05, FK81, FKS89}. The adversarial planted coloring model therefore generalizes
random planted models, in which the host graph is drawn at random and/or the planted partitioning is chosen uniformly at random. 

\paragraph{Our contribution.}
We give the first sublinear-time algorithm for the planted coloring problem in this model. Our result works in the adversarial planted model, and therefore also applies to the corresponding random planted models. Specifically, in the adversarial planted $k$-coloring model, we construct a coloring oracle with $\widetilde O(n^{1/2+O(1/\log(d/\lambda)})$ preprocessing time and space, where $n$ is the number of vertices. The oracle answers each color query in $\widetilde O(n^{1/2+O(1/\log(d/\lambda))})$ time and, with high constant probability, agrees with the planted coloring on at least a $0.99$ fraction of the vertices.
\begin{samepage}
\begin{thm}\label{thm:main-result}
Let \(k \geq 3\) be a constant, let $c_k$ 
be a sufficiently  small constant (depending on $k$), and let $C_k$ be a sufficiently large constant (depending only on $k$). For every $C_k < d<n $ and every $\lambda \leq c_k d$, there exists a randomized algorithm for the adversarial planted \(k\)-coloring model
(Definition~\ref{def:model}) that constructs a coloring oracle using
\(\widetilde O\!\left(n^{1/2+O(1/\log(d/\lambda))}\right)\) preprocessing time and space, where 
$n$ denotes the number of vertices in the graph. For every query vertex $v$, the oracle outputs a color in $[k]$ in time \(\widetilde O\!\left(n^{1/2+O(1/\log(d/\lambda))}\right)\), and its answers are consistent across queries. With probability at least $0.9$, the oracle agrees with the planted coloring (up to a permutation of the colors) on all but at most $ O\!\left(\sqrt{\frac{\lambda }{d}}n\right)$ vertices. 
\end{thm}
\end{samepage}

In particular, the running time is $\approx n^{1/2 + o(1)}$ when $\lambda/d = o(1)$. For dense graphs with $d = n^{\delta}$ and $\lambda = O(\sqrt{d})$ for a constant $\delta$, the running time becomes  $\widetilde O(n^{1/2})$. 

As in prior works on planted coloring, we treat the number of colors \(k\) as a constant. This assumption is used both in the spectral concentration bounds and in our sublinear implementation (see Remark~\ref{rem:growing-k} for more details). 

Our algorithm may be viewed as a sublinear-time implementation of the classical spectral algorithm of \cite{AK97}. 
The key technical ingredient is providing sublinear time dot-product access to the bottom eigenspace of the adjacency matrix, which we achieve by adapting recent techniques from spectral clustering  \cite{KKW}. 

\paragraph{Related work on planted coloring.} Blum and Spencer \cite{BlumSpencer} and Alon and Kahale \cite{AK97} studied the problem of planted $k$-coloring in random graphs, and showed that exact recovery is possible in polynomial time. \cite{DF16} studied the role of randomness in the planted model, and showed that exact recovery is possible if the coloring is planted at random in an adversarial spectral expander, and that approximate recovery is possible even when the coloring is adversarial. \cite{Kumar17} proved that approximate recovery is possible in low threshold rank graphs for pseudorandom colorings. Recently, \cite{BHK25, Buhai25, Hsieh26} further considered the setting in which the coloring is planted in a one-sided expander or a low threshold rank graph. 

\paragraph{Related work on sublinear spectral algorithms.}
Our techniques are closely related to work on sublinear spectral clustering \cite{CzumajPS15, CKK18, GKLMS21, SP23, FKKLMW26, KKW}. Most relevant to our work is \cite{KKW}, which uses a low-degree polynomial of the random walk matrix to provide access to the top eigenspace of the adjacency matrix under a clusterability assumption. Although this setting differs from ours, we adopt their general approach.
\cite{peng2023sublinear} and \cite{Maxcut24} applied spectral techniques to approximate max-cut in sublinear time on expanders and clusterable graphs. While these works rely on structural properties of the bottom eigenspace in their analysis, the algorithms themselves do not provide direct access to this eigenspace. The work of \cite{peng2023sublinear} also showed that testing whether an expander is 3-colorable versus far from $3$-colorable requires \(\Omega(n)\) time. We note that the hard instances in  \cite{peng2023sublinear} are constant degree expanders with expansion $\phi = \Omega(1/d)$, whereas our results require stronger spectral expansion. 

\section{Technical Overview}\label{sec:tech}
Our algorithm can be viewed as a sublinear implementation of the spectral approaches of Alon and Kahale \cite{AK97} and David and Feige \cite{DF16}. These works show that planted colorings can be approximately recovered from the bottom eigenspace of the adjacency matrix. The main difficulty is that computing this exactly requires reading the whole graph and computing the eigenvectors. 

Recent progress on sublinear spectral algorithms has focused on approximating inner products of spectral embeddings associated with the \emph{top} eigenspace of a graph, typically under strong clusterability assumptions
\cite{GKLMS21, KKW}. Our main technical contribution is to extend this line of work to the setting of graph coloring by providing sublinear-time access to the \emph{bottom} eigenspace of the adjacency matrix. While bottom eigenvectors have long played a central role in polynomial-time spectral algorithms, sublinear-time methods for accessing this part of the spectrum have remained largely undeveloped. 
\paragraph{The spectral approach of \cite{AK97} and \cite{DF16}.}
We start by reviewing the classical approach of \cite{AK97} and \cite{DF16}. These works showed that an approximate coloring can be recovered
from the bottom eigenspace of the adjacency matrix of the planted input graph. Specifically, the bottom $k-1$ eigenvalues of the normalized adjacency matrix have value $-\frac{1}{k-1}$ $ \pm O(\sqrt{\lambda/d})$, while all remaining
eigenvalues (except for the trivial top eigenvalue \(1\)) are bounded in absolute value by $O(\sqrt{\lambda/d})$ 
\footnote{\cite{DF16} analyze the unnormalized adjacency matrix. Since our algorithm is based on random walks, it is more convenient for us to work with the normalized adjacency matrix. An analogous statement for the normalized adjacency matrix \(\bar A\) follows by a similar argument (see Lemma~\ref{lemma:C4}).}. The coloring can be approximately recovered from the spectral embedding in this bottom eigenspace. 

\begin{definition}[Spectral embedding]\label{def:spec-emb} Let $\bar A = D^{-1/2}AD^{-1/2}$ denote the normalized adjacency matrix of the input graph $G'$, and let $\U \in \mathbb{R}^{n\times (k-1)}$ denote the matrix whose columns are the bottom $k-1$ eigenvectors of $\bar A$. For every vertex $v \in V$, the spectral embedding of $v$, denoted $f_v \in \mathbb{R}^{k-1}$, is the $v$-th row of $\U$, i.e.,  
$$f_v \coloneqq \U^\top   \1_v.$$
\end{definition}
\begin{remark}
The spectral embeddings $f_v$ are not uniquely defined. However, the dot products $\langle f_u, f_v \rangle $, and hence also the Euclidean distances $\|f_u - f_v\|_2$, are uniquely defined (see \Cref{rem:unique_embedding}). 
\end{remark}

As shown in~\cite{DF16}, spectral embeddings of vertices belonging to the same color class are close to each other, while spectral embeddings of vertices belonging to different color classes are far from each other. Specifically, for all but an \(O(\lambda/d)\) fraction of vertices, the following holds\footnote{\Cref{eq:tech_distanec} can be extracted from the analysis
of \cite{DF16} for embeddings of the \emph{unnormalized} adjacency
matrix, although it is not stated there explicitly in this form. We
prove an analogous guarantee for the normalized adjacency matrix
(see \Cref{lemma:B_bound}).}:
\begin{equation}\label{eq:tech_distanec}
  \|f_u - f_v\|_2^2 =
\begin{cases}
\le \dfrac{k}{70n}, & \text{if } \chi(u) = \chi(v), \\[6pt]
\ge \dfrac{k}{2n},  & \text{if } \chi(u) \neq \chi(v).
\end{cases}  
\end{equation}

By \Cref{eq:tech_distanec}, approximate recovery of the planted coloring can be achieved using the following two steps: 
\begin{enumerate}
    \item Identify a set of \emph{color representatives} \(c_1 \in C_1, \dots, c_k \in C_k\), i.e., one
    vertex from each color class.
    \item Given a query vertex \(u\), output the color of the closest color representative 
    \[
    \Color(u) = \arg\min_{i \in [k]} \|f_u - f_{c_i}\|_2^2 .
    \]
\end{enumerate}
The main challenge is to perform this in sublinear time, without explicitly computing the
full spectral embedding. 

\paragraph{Spectral embeddings via polynomials.} 
In order to approximate inner products $\langle f_u, f_v\rangle$ between spectral embeddings (and hence also the Euclidean distances), we will approximate the relevant projection matrix by  a low-degree polynomial in the random walk matrix. The inner products can then be approximated by collision counting. This idea was used in~\cite{KKW} in the context of sublinear-time spectral clustering, and we adapt their general approach. 
Let \(\bar A = U \Sigma U^\top\) denote the eigendecomposition of the normalized adjacency matrix \(\bar A\). 
We are interested  in approximating
\[
\langle f_u,f_v\rangle
= \1_u^\top U_{[k-1]} U_{[k-1]}^\top \1_v ,
\]
where  $U_{[k-1]}U_{[k-1]}^\top$ is the projection onto the eigenspace corresponding to eigenvalues with value $\approx -1/(k-1)$. 

To this end, we seek a polynomial
\[
p(x) = \sum_{t \ge 0} c_t x^t
\]
\begin{samepage}
of the form \(p(x) = x^t q(x)\), where $ t = \Omega\!\left(\frac{\log n}{\log(d/\lambda)}\right)$, 
such that the coefficients \(c_t\) are small and the following properties hold:
\begin{itemize}
    \item \(p(x) \approx 1\) when \(\left|x + \frac{1}{k-1}\right|
    \le O\!\left(\sqrt{\frac{\lambda}{d}}\right)\), and
    \item \(p(x) \approx 0\) when \(|x| \le O\!\left(\sqrt{\frac{\lambda}{d}}\right)\) or \(x=1\).
\end{itemize}
\end{samepage}
With such a polynomial, the matrix
\[
p(\bar A) = U p(\Sigma) U^\top
\]
acts as an approximate projector onto the bottom $k-1$ eigenvectors. As a result, 
\begin{align*}\label{eqn:tech_poly}
    \langle f_u, f_v\rangle  &= \1_u^\top U_{[k-1]}U^\top _{[k-1]}\1_v  \\
    & \approx \1_u^\top  U p(\Sigma)^2U^\top  \1_v  = \left \langle p(\bar A )\1_u, p(\bar A)\1_v\right  \rangle = \left \langle  \sum_t c_t \bar A ^t \1_u ,  \sum_t c_t  
\bar A^t \1_v\right \rangle \\
&= \sum_{t,t'}c_t c_{t'} \langle  \bar A ^t \1_u , \bar A ^{t'} \1_v\rangle   . 
\end{align*}

Each inner product \(\langle \bar A^t \1_u, \bar A^{t'} \1_v\rangle\) can be estimated
efficiently using random walks. Indeed, let \(M = AD^{-1}\) denote the random walk transition matrix. For $t = \Omega\!\left(\frac{\log n}{\log(d/\lambda)}\right)$, inner products of the form
\(\langle M^t \1_u, M^{t'} \1_v\rangle\) can be approximated using
$\approx \sqrt{n}$ random walks via standard collision-counting 
arguments. Using the relation \(\bar A = D^{-1/2} M D^{1/2}\), inner products involving powers of \(\bar A\) can be approximated similarly. 

It is crucial that the coefficients $c_t$ of the polynomial~$p$ are small:  each term
\(\langle \bar A^t \1_u, \bar A^{t'} \1_v\rangle\) must be estimated to precision inverse polynomial in \(|c_t c_{t'}|\), and hence the overall running time depends polynomially on the size of coefficients. 

\cite{KKW} used a similar idea to cluster graphs. In their setting, the input graph is assumed to be clusterable, meaning that the top $k$ eigenvalues lie in the interval $[1-\epsilon,1]$ for some error parameter $\epsilon$, while all remaining eigenvalues are bounded away from $1$. The goal is to approximate embeddings in this \emph{top} eigenspace of the graph, which they achieve by constructing a polynomial that approximates the indicator on $[1-\epsilon, 1]$.

A key challenge with adapting their construction to our setting is that the coefficients of their polynomial grow superpolynomially in their error parameter $\epsilon$, and can be as large as $\approx (1/\epsilon)^{O(\log(1/\epsilon))} n^{O(\epsilon)}$. Since the running time depends polynomially on the coefficient size, the algorithm is sublinear only for relatively large values of $\epsilon$.

In our setting, the relevant analogue of the error parameter is $\sqrt{\lambda/d}$, and a superpolynomial dependence on $\sqrt{d/\lambda}$ would lead to too large running times. For example, for random \(d\)-regular graphs with \(d = n^{\Omega(1)}\), this would lead to running times that are superpolynomial in $n$. 

To overcome this issue, we modify the polynomial construction of \cite{KKW} and give a more careful analysis of the resulting polynomial. This gives a polynomial that approximates an indicator function around $-1/(k-1)$, while having coefficients with absolute value at most $n^{O\left(1/\log(d/\lambda)\right)}$. In particular, the coefficient size, and hence the running time, decreases as $\lambda/d$ decreases.
\begin{thm}\label{cor:polynomial}
Suppose that $\lambda/d$ is less than a sufficiently small constant depending on $k$. Let $\delta, C> 0$ be constants. For every $\frac{10C \log n}{\log(d/\lambda)} \leq t \leq O\left( \frac{\log n}{\log(d/\lambda) }\right)$, there exists a polynomial~$p$ of the form $p(x) = x^t q(x)$, where $q$ is a polynomial of degree $O(\sqrt{\lambda/d} \log n)$, such that the coefficients of~$p$ are bounded in absolute value by $ n^{O\left(\frac{1}{\log(d/\lambda)}\right)}$, and
\begin{enumerate}[label=(\textbf{\arabic*})]
    \item $\left |p(x)-1\right| \leq \delta $ for $x \in \left[-\frac{1}{k-1} - C\sqrt{\frac{\lambda}{d}}, -\frac{1}{k-1} + C\sqrt{\frac{\lambda}{d}}\right]$ \label{con:poly1}
    \item $|p(x)|\leq  n^{-2}$  for   $|x|\leq C \sqrt{\frac{\lambda}{d}}$.  \label{con:poly0special}
\end{enumerate}
\end{thm}
In fact, our proof establishes a more general result: for any constant $\alpha \in [-1,1]$ with $\alpha \neq 0$ and any $0 < \epsilon\ll |\alpha|$, we construct a polynomial approximating an indicator on $[\alpha-\epsilon,\alpha+\epsilon]$ with similar coefficient bounds (see \Cref{thm:polynomial}). The proof of \Cref{cor:polynomial} is presented in \Cref{sec:polynomial}. 

While this polynomial successfully maps the eigenvalues near $-1/(k-1)$ to values close to $1$ and eigenvalues of size at most $\sqrt{\lambda/d}$ to values close to $0$, it does not control the value of $p(1)$, which corresponds to the top eigenvalue. We handle this issue by explicitly removing the contribution of the top eigenvector. Since the top eigenvector of \(\bar A\) is
\(e_1 = D^{1/2}\1/\sqrt{2m}\), its contribution can be computed exactly, and we can therefore apply the polynomial to the matrix \(\bar A - e_1 e_1^\top\), in which the top
eigenvalue has been shifted to zero.

The algorithms for approximating $\langle f_u, f_v\rangle $ (\Cref{alg:SDP}) and $\|f_u - f_v\|^2_2$ (\Cref{alg:approxdist}) are presented in  \Cref{sec:alg}. Their guarantees (\Cref{lemma:SDP} and \Cref{lemma:approxdist}) are proved in \Cref{sec:SDP} and \Cref{sec:apprxdst}. 

\paragraph{Finding color representatives.}
Given access to pairwise distances between spectral embeddings, it remains to pick one representative vertex from each planted color class. The polynomial-time approach of~\cite{DF16} finds such representatives by repeatedly sampling $k$ random vertices until exactly one comes from each class. However, their method for verifying whether a given sample is valid requires $\Omega(n)$ time, which is too costly in our setting. We instead sample $\Theta(k\log k)$ vertices and build a small ``similarity graph'' on the sample, connecting two sampled vertices when their embeddings are similar. This is similar to approaches used for clustering~\cite{CzumajPS15, SP23}. 

The similarity graph decomposes into exactly $k$ connected components (one per color class) with  high constant probability, so choosing one vertex from each component yields the required representatives. We present the algorithm for obtaining color representatives in \Cref{sec:prepro}.

\section{Preliminaries and Notation}\label{sec:prelim}
We start by formally defining the graph access model and the input model. 
\paragraph{Graph access model.} We work in the adjacency-list query model. In this model, an algorithm can perform neighbor queries, i.e., for the $i$-th neighbor of any vertex, and degree queries, i.e., for the degree of any vertex. We additionally assume that the number of edges $m$ is given as part of the input to the algorithm. 

\begin{definition}[$\lambda$-expander]
Let $\lambda_1(A_G) \geq \lambda_2(A_G)  \geq \dots \geq \lambda_n(A_G)$ denote the eigenvalues of the unnormalized adjacency matrix $A_G$ of a graph $G$. Say that $G$ is a $\lambda$-expander if $\max\{|\lambda_2(A_G)|,|\lambda_n(A_G)|\}\le \lambda$.
\end{definition}

\begin{definition}[Adversarial planted \(k\)-coloring]\label{def:model}
Let $k \geq 3$ be a constant.
An adversary chooses an \(n\)-vertex \(d\)-regular \(\lambda\)-spectral expander \(G=(V,E)\) together with a balanced partition
\(\mathcal C = (C_1,\dots,C_k)\) of \(V\), such that $|C_i| = n/k$ for all $i \in [k]$ and such that every vertex has at least $\Omega(d)$ neighbors outside its color class. The partition \(\mathcal C\) defines a \emph{planted \(k\)-coloring} \(\chi : V \to [k]\) by setting \(\chi(v)=i\) if \(v \in C_i\).

The planted instance is obtained by deleting all monochromatic edges from \(G\). That is, let
\[
E' \coloneqq \{(u,v) \in E : \chi(u) \neq \chi(v)\},
\]
and let \(G' = (V,E')\) denote the resulting \(k\)-colorable graph.
Write \(G' = P_{\mathcal C}(G)\) when $G'$ is obtained this way from the graph $G$ and the partition $\mathcal C$. 
\end{definition}

\begin{remark}[Model assumptions]
The model above is almost identical to the planted adversarial model of \cite{DF16}, except that we impose the additional condition that every vertex has at least a small constant fraction of its neighbors outside of its own color class. This ensures that after deleting monochromatic edges, every vertex in $G'$ has degree $\Omega(d)$. Since the typical degree in $G'$ is
$(1-\frac{1}{k})d$, this is a mild restriction that only rules out vertices that are almost isolated by the adversary. In particular, if the color classes are selected \emph{at random} and $d = \Omega(\log n)$, then this holds with high probability. The purpose of this assumption is
to ensure that the normalized adjacency matrix of $G'$ has the required
spectral properties. Our sublinear-time algorithm is based on random walks and therefore relies on spectral properties of the normalized adjacency matrix. In contrast, prior polynomial-time spectral algorithms for planted $k$-coloring work directly with the \emph{unnormalized} adjacency matrix, and therefore do not need this additional balanced degree condition.

For simplicity of presentation, Definition~\ref{def:model} assumes that the planted coloring is exactly balanced, i.e.,
$|C_i| = n/k$ for all $i \in [k]$.
Our results extend to approximately balanced planted colorings satisfying $|C_i|/|C_j| = O(1)$ for all $i,j \in [k]$. The exact balance assumption is consistent with the standard formulation in prior work (e.g., \cite{BlumSpencer, AK97, DF16}).

Similar to prior work, we assume that $k=O(1)$, and suppress the dependence on $k$ in the $O$-notation.  
\end{remark}
\begin{remark}[Dependence on the number of colors]\label{rem:growing-k} The assumption that \(k\) is constant is used in two crucial ways. First, it is used to show that the relevant eigenspace carries useful information about the planted coloring. In the planted \(k\)-coloring model, the relevant eigenspace corresponds to eigenvalues near $-\frac{1}{k-1}$, while the other nontrivial eigenvalues are controlled up to \(O(\sqrt{\lambda/d})\) in the normalized adjacency matrix. Thus, with the currently known bounds, separating the relevant eigenspace from the rest of the spectrum requires a separation of the form $\sqrt{\lambda/d} \ll 1/k$. 

Moreover, for nonconstant \(k\), existing spectral bounds (such as  \Cref{lemma:Bspect}, or  Lemma B.16 in \cite{DF16}) incur additional dependence on \(k\). In particular, they would give a bound of 
  $  O\left(k\sqrt{\lambda/d}n\right)$ on the number of vertices whose spectral embedding significantly deviates from the rest of their color class. These issues also apply to previous polynomial-time spectral algorithms for planted coloring \cite{AK97,DF16}.

Second, the constant $k$ assumption is used in our sublinear implementation, since our polynomial has to separate the eigenvalues near $-\frac{1}{k-1}$ from those near $0$. As $k$ grows, this is increasingly hard, which is reflected in the coefficient size, and hence the sampling complexity. 

Thus, extending the result to nonconstant \(k\) would require stronger spectral concentration bounds, as well as a polynomial construction whose coefficients remain controlled for nonconstant \(k\).
\end{remark}

The goal is to construct a coloring oracle, as defined below. 
\begin{definition}[Coloring oracle]\label{def:oracle}
A coloring oracle is a randomized data structure \(\mathcal O\) that, given query access to a graph \(G'\) from the adversarial planted \(k\)-coloring model, provides consistent query access to a function 
$ \Color: V \rightarrow [k]$.
 
Say that \(\mathcal O\) defines an \(\varepsilon\)-approximate coloring if, with probability at least
\(0.9\) over the random bits of~\(\mathcal O\), there exists a permutation \(\pi : [k] \to [k]\) such that
\[
\left|\left\{ v \in V : \Color(v) \neq \pi(\chi(v) \right)\}\right|
\le \varepsilon n ,
\]
where $\chi$ is the planted coloring. 
\end{definition}

\paragraph{Notation.} Given the input graph $G' = (V,E')$, let $m$ denote the number of edges $|E'|$. Write $A$ for the unnormalized adjacency matrix of $G'$ and $D$ for the diagonal degree matrix. Write $\bar A = D^{-1/2}AD^{-1/2}$ for its normalized adjacency matrix, and write $M = AD^{-1}$ for the random walk matrix. Observe that
$$ M = D^{1/2}\bar A D^{-1/2}.$$

Let $\lambda_1(\bar A) \ge \lambda_2(\bar A) \ge \cdots \ge \lambda_n(\bar A)$ denote the eigenvalues of $\bar A$ , and let $e_1(\bar A), \ldots, e_n(\bar A)$ denote a corresponding orthonormal basis of eigenvectors. When it is clear that we are talking about the matrix $\bar A$, we omit the dependence on $\bar A$ from the notation.

For a vertex $v \in V$, write $\1_v \in \mathbb{R}^n$ for the indicator of $v$, i.e., the vector with entry $1$ at index $v$ and $0$ everywhere else. Write $\1 \in \R^n$ for the vector with all entries equal to $1$. 
\begin{obs}\label{obs:e1}
     $\bar A D^{1/2}\1 = D^{1/2}\1. $ 
    In particular, 
    $$ e_1(\bar A) =  \frac{D^{1/2}\1}{\| D^{1/2}\1\|_2} = \frac{D^{1/2}\1}{\sqrt{2m}}\qquad \text{and} \qquad \lambda_1(\bar A)  =1.$$
\end{obs}
For a polynomial $q$, write $\deg(q)$ to denote the degree of $q$ and $\max|\coeff(q)|$ to denote the absolute value of the largest (in absolute value) coefficient of $q.$  

We use $\widetilde O$ notation to suppress \(\operatorname{polylog} n\) factors. 

\section{The Algorithm}\label{sec:alg}
In this section, we present the algorithm. The algorithm \SDP{} (\Cref{alg:SDP}) approximates inner products $\langle f_u, f_v\rangle$ of spectral embeddings. It runs $\approx \sqrt{n}$ random walks of various lengths starting from $u$ and $v$, and then combines the resulting collision counts with the coefficients of the polynomial $p$ from \Cref{cor:polynomial}. 

To change from the random walk matrix $M = AD^{-1}$ to the normalized adjacency matrix $\bar A =  D^{-1/2} M D^{1/2}$, we appropriately reweigh collisions according to the matrix $D^{-1/2}$. Finally, we subtract the term $\frac{1}{2m}$ in order to cancel out the contribution of the top eigenvalue.

\begin{algorithm}[H]
    \caption{$\randomwalks(r, t, u)$}\label{alg:randwalks}
    \begin{algorithmic}
        \State Run $r$ random walks of length $t$ starting from $u$
        \State Let $\hat{p}_{u}(v)$  be the fraction of random walks that end at $v$  \Comment{vector $\hat{p}_u$ has support at most $r$ 
        \State \Return{ $\hat{p}_{u}$}}
    \end{algorithmic}
\end{algorithm}

 \begin{algorithm}[H]
\caption{\SDP$(u,v)$}
\label{alg:SDP}
\begin{algorithmic}[1]
\State $t_{\min} \gets\ O \left( \frac{\log n}{\log(d/\lambda)}\right)$  \Comment{shortest random walk length}
\State $t_{\Delta} \gets O(\sqrt{\lambda/d} \log n)$  \Comment{number of different walk lengths}
\State $r \gets \widetilde O\left( n^{1/2+ O\left(\frac{1}{\log(d/\lambda)}\right)}\right)$\Comment{number of random walks}

 \For{$t = t_{\min}, \dots t_{\min} + t_{\Delta}$}
\State $\hat{p}^t_u \gets$ \randomwalks$(r, t, u)$ 
 \Comment{vector $\hat{p}^t_u$  has support at most $r$}\label{line:sdp_rwu}
 \State $\hat{p}^t_v \gets$ \randomwalks$(r, t, v)$ 
\Comment{vector $\hat{p}^t_v$  has support at most $r$\label{line:sdp_rwv}}
 \EndFor
 \State  \Return{$\sqrt{\deg(u)}\sqrt{\deg(v)}\sum_{t,t'} c_tc_{t'} \left(\left \langle D^{-1/2}  \hat{p}^t_u , D^{-1/2} \hat{p}^{t'}_v\right \rangle - \frac{1}{2m}\right)$}\label{line:sdp_return}
  \State \Comment{ $c_t$ is the coefficient of $x^t$ in the polynomial $p(x)$ from \Cref{cor:polynomial}} 
\end{algorithmic}
\end{algorithm}

\SDP{}  has the following guarantee: outside a small deterministic set \(B\) of ``bad" vertices, $\SDP(u,v)$ approximates $\langle f_u, f_v\rangle $ up to an additive error of $O(1/n)$. The set \(B\) is defined formally in \Cref{def:goodbad} and consists of vertices whose degree or spectral embedding is atypical relative to their planted color class. We bound the size of $B$ in  \Cref{lemma:B_bound}.

\begin{lemma}[Correctness of \SDP]\label{lemma:SDP}
The procedure \SDP{} (Algorithm~\ref{alg:SDP}) runs in time $\widetilde O\left(n^{1/2 + O\left(\frac{1}{\log(d/\lambda)}\right)}\right)$. For every $ u,v \in V \setminus B$, 
$$ \Pr\left[\left |\SDP(u,v) - \langle f_u, f_v\rangle \right| \leq \frac{1}{80}\frac{k}{n} \right]  \geq 0.9,$$
where the probability is over the internal randomness of $\SDP$. 
\end{lemma}
\noindent 
The correctness guarantee of \SDP{} (\Cref{lemma:SDP}) is proved in \Cref{sec:SDP}. 

To compute distances $\|f_u - f_v\|^2_2$, we use the algorithm \approxdist{} (\Cref{alg:approxdist}). The algorithm boosts the success probability of \SDP{} and approximates  $\|f_u - f_v\|^2_2  = \|f_u\|^2_2 + \|f_v\|^2_2 - 2\langle f_u, f_v\rangle$ by approximating each of the terms. 

 \begin{algorithm}[H]
\caption{\approxdist$(u,v)$}
\label{alg:approxdist}
    \begin{algorithmic}[1]
 \For{$i = 1, \dots O(\log n)$}
   \State  ${\|f_u\|_{\apx,i}^2}\gets \SDP(u,u)$ \label{line:dotu} \Comment{Approximate $\|f_u\|^2_2$}
    \State $\|f_v\|_{\apx,i}^{2} \gets \SDP(v,v)$\label{line:dotv}\Comment{Approximate $\|f_v\|^2_2$}
    \State $\langle f_u, f_v\rangle_{\apx,i}  \gets \SDP(u,v)$\label{line:dotuv} \Comment{Approximate $\langle f_u, f_v\rangle $}
    \State $\|f_u - f_v\|_{\apx,i}^{2}\gets \|f_u\|_{\apx,i}^{2}  + \|f_v\|_{\apx,i}^{2} - 2\langle f_u, f_v\rangle_{\apx,i}  $ \label{line:disti}
 \EndFor
\State \Return $\median_i (\|f_u - f_v\|_{\apx,i}^{2})$ \label{line:distreturn}
\end{algorithmic}
\end{algorithm}
\begin{lemma}[Correctness of \approxdist]\label{lemma:approxdist}
The procedure \approxdist{} (Algorithm~\ref{alg:approxdist}) runs in time $\widetilde O\left(n^{1/2 + O\left(\frac{1}{\log(d/\lambda)}\right)}\right)$. With probability at least $1-n^{-100}$, 
$$  \left |\approxdist(u,v) - \| f_u- f_v\|_2^2 \right| \leq \frac{k}{20n} $$
for all $ u,v \in V \setminus B$,  where the probability is over the internal randomness of $\approxdist$. 
\end{lemma}
As a corollary, $\approxdist{(u,v)}$ can be used to determine whether $u$ and $v$ belong to the same color class, for all $u,v \in V \setminus B$. 
\begin{samepage}
\begin{cor}\label{cor:approxdist_decision}
With probability at least $1-n^{-100}$, the following holds for all $u,v \in V \setminus B$: 
\begin{itemize}
    \item If $\chi(u) =  \chi(v)$, then $\approxdist(u,v) \leq  \frac{k}{4n}$. 
   \item If $\chi(u) \neq  \chi(v)$, then $\approxdist(u,v) \geq \frac{k}{2n}$. 
\end{itemize}
\end{cor}
\end{samepage}
\Cref{lemma:approxdist} and \Cref{cor:approxdist_decision} are proved in \Cref{sec:apprxdst}. 
In particular, \Cref{alg:approxdist} can be used as a coloring oracle: given a set of color representatives $c_1,\dots,c_k$, where each $c_i$ belongs to a different color class and none of them lies in the set of bad vertices $B$, one can answer each query $u$ by assigning the color of its nearest neighbor among the $c_i$. \Cref{alg:color_main}  does this. 

\begin{algorithm}[H]
\caption{\Color$\left(u, \{c_i\}_{i\in [k]}\right)$
\newline\textbf{Input:} vertex $u$ and set of color representatives $\{c_i \}_{i \in [k]}$ (see \Cref{def:rep})
\newline \textbf{Output:} Index $i \in [k]$ such that  $\chi(u) =\chi(c_i)$ 
}
\label{alg:color_main}
    \begin{algorithmic}[1]
\State \Return $\argmin_{i \in [k]} \approxdist(u,c_i)$
\end{algorithmic}
\end{algorithm}

\subsection{Obtaining color representatives} \label{sec:prepro}
It remains to find a valid set of color representatives. 
The procedure \prepro{} (\Cref{alg:app_centers}) below achieves this. The algorithm samples \(\Theta(k \log k)\) vertices uniformly at random, which ensures that, with high probability, the sample contains at least one vertex from each color class. The algorithm then constructs a ``similarity graph" on the sampled vertices, connecting two sampled vertices if their distance is sufficiently small. By the clustering properties of the spectral embeddings (see \Cref{lemma:B_bound}), with a good probability, vertices from the same color class form a clique in this similarity graph, while vertices from different color classes are disconnected. Finally, the algorithm selects one arbitrary vertex from each connected component of the similarity graph and discards the rest. We show that, with probability at least \(0.95\), this procedure returns exactly
one representative from each color class, yielding a valid set of representatives for use in the coloring oracle.
\Cref{alg:app_centers} is very similar to the preprocessing routines of \cite{SP23} and \cite{FKKLMW26}. 

\begin{algorithm}[H]
\caption{\prepro{}}\label{alg:app_centers}
\begin{algorithmic}[1]
\State $S \gets$  Multiset of $10k\log k$ vertices sampled uniformly at random from $V$
\State $H \gets (S, \emptyset)$\label{line:init_H} \Comment{$H$ will be the same-color graph on $S$ (see \Cref{def:same-color-graph})}
\ForAll{$u, v \in S$}
    \If{$\approxdist(u,v) \leq \frac{k}{3n}$}
        \State Add edge $\{u, v\}$ to $H$ \label{line:add_to_H}
    \EndIf
\EndFor\label{line:end_init_H}
\State $i \gets 1$
\While{$V(H) \neq \emptyset$}\label{line:begin_centers}
    \State Select arbitrary vertex $u$ from $V(H)$\label{line:select_center} \Comment{Select a representative of a color class}
    \State $c_i \gets u$
    \State $T_i \gets \{u\} \cup N_H(u)$ \Comment{The sampled vertices from the same color class as $u$}
  \State $H \gets H[V(H)\setminus T_i]$
\Comment{Remove all vertices from the color class of $u$}
    \State $i \gets i + 1$
\EndWhile\label{line:end_centers}
\State \Return $c_1, c_2, \dots  $
\end{algorithmic}
\end{algorithm}

\begin{samepage}
\begin{definition}[Color Representatives]\label{def:rep}
A set of vertices $\{c_i\}_{i \in [k]}$ is called a \emph{set of color representatives} if the following conditions hold:
\begin{enumerate}
    \item Every $c_i$ is a good vertex (see \Cref{def:goodbad}).
    \item The vertices have distinct colors, i.e., $\chi(c_i) \neq \chi(c_j)$ for all $i \neq j$.
\end{enumerate}
\end{definition}
\end{samepage}

\begin{restatable}[Correctness of \prepro{}]{lemma}{preprolemma}\label{lemma:prepro_main} 
        The procedure \prepro{}  (\Cref{alg:app_centers}) runs in time $\widetilde O(n^{1/2 + O(1/\log(d/\lambda))})$. With probability at least $0.95$, it returns a set of $k$ color representatives (as per \Cref{def:rep}). 
\end{restatable}

The proof of \Cref{lemma:prepro_main}  follows the analysis of the preprocessing procedure in $\cite{FKKLMW26}$, and is presented in Appendix \ref{sec:prepro_analysis}. 

\section{Analysis}\label{sec:analysis}
In this section, we analyze the algorithms presented in \Cref{sec:alg} and prove the main result, \Cref{thm:main-result}. The section is organized as follows. In \Cref{sec:properties}, we formally define the set $B$ of bad vertices and give a bound on its size. In \Cref{sec:SDP}, we prove  \Cref{lemma:SDP}, establishing the correctness of \SDP{}. Next, in \Cref{sec:apprxdst}, we prove \Cref{lemma:approxdist} and \Cref{cor:approxdist_decision}, thereby establishing the correctness of \approxdist{}. In \Cref{sec:polynomial} we construct the polynomial and prove \Cref{cor:polynomial}. Finally, in \Cref{sec:proof_main}  we combine everything to prove \Cref{thm:main-result}. 

\subsection{Properties of the input instance}\label{sec:properties}
We start by stating some of the key properties of the input instance $G'$. The proofs of these properties are deferred to Appendix \ref{sec:properties_pfs}. 

 \begin{obs}\label{cor:m}
By the model assumption (\Cref{def:model}), $\deg(v) = \Theta(d)$ for all $v \in G'$ and  $\Omega(nd) \leq m \leq nd$. 
 \end{obs}

\begin{restatable}[Normalized version of Lemma C.4 in \cite{DF16}]{lemma}{spectralprops}\label{lemma:C4}
      Let $G$ be a $d$-regular $\lambda$-expander, where $\lambda \leq c_k d$ (for some constant $c_k$), and let $G' = P_{\mathcal C}(G)$ (as per \Cref{def:model}). Then
    \begin{enumerate}[label=(\textbf{\arabic*})]
    \item The eigenvalues of the normalized adjacency matrix $\bar A$ of $G'$ have the following spectrum: \label{item:C4_spectrum}
    \begin{enumerate}
    \item $\lambda_1(\bar A) = 1$ \label{item:lambda1}
        \item $\left| \lambda_{i}(\bar A) + \frac{1}{k-1}\right| \leq O\left(  \sqrt{\frac{\lambda}{d}}\right)$ for $n-k+2 \leq i \leq n$. \label{item:lambdatop}
        \item $|\lambda_i(\bar A)| \leq O\left( \sqrt{\frac \lambda d}\right) $ for $2 \leq i \leq n-k+1$ \label{item:lambdamid}
    \end{enumerate}
    \item \label{item:projproperty}
    For every $x \in \spn(\{ \1_{C_i}\}_{i \in [k]})$ such that $\langle x, \1 \rangle = 0$, there exists a vector $x^* \in \spn(\{ e_i(G')\}_{i \in [n-k+2, n]})$ such that
     $  \|{D^{1/2}x}- x^*\|_2^2=  O\left ( \sqrt{\frac{\lambda}{d}} \right)\|D^{1/2}x\|^2_2$. Here $C_1, \dots C_k$ denotes the planted partition into color classes. \label{item:C_4_uv} 
    \end{enumerate}
\end{restatable}
  
\cite{DF16} proves analogous guarantees to those in \Cref{lemma:C4} for the eigenvalues and eigenvectors of the \emph{unnormalized} adjacency matrix (see Lemma C.4 in \cite{DF16}). Our proof of \Cref{lemma:C4} is similar and is included in Appendix \ref{sec:C4} for completeness. 
\begin{remark}\label{rem:unique_embedding}
It follows from guarantee \ref{item:C4_spectrum} in \Cref{lemma:C4} that the space spanned by the bottom $k-1$ eigenvectors of $\bar A$ is uniquely defined, i.e., the choice of $\U$ is unique up to multiplication by an orthogonal matrix $R\in \mathbb{R}^{(k-1) \times (k-1)}$ on the right. 

While the choice of $f_v$ for $v \in V$ is not unique, the dot product 
	between the spectral embeddings of $u\in V$ and $v\in V$ is well defined, since for every orthogonal 
	$R\in \mathbb{R}^{(k-1)\times (k-1)}$ one has  
	\[\langle R^\top f_u, R^\top f_v\rangle=(R^\top f_u)^\top (R^\top f_v)=f_u^\top  R R^\top f_v=\langle f_u, f_v \rangle\text{.}\]
\end{remark}
A useful property of the spectral embeddings is that they concentrate around their color class center. 
\begin{definition}[Color class center]\label{def:mu}
For each color class $C_i \in \mathcal C$, let $\mu_i$ denote the \emph{color class center}
\begin{equation*}
\mu_i \coloneqq \frac{1}{|C_i|} \sum_{v \in C_i} f_v = \frac{1}{|C_i|} \U^\top \1_{C_i}. 
\end{equation*}
\end{definition}
We define a set of \emph{good} vertices consisting of those vertices that are close to their own color class center, far from other color class centers, and whose degree is close to the ``typical" value $\frac{d(k-1)}{k}$.
\begin{definition}[Good vertex]\label{def:goodbad}
    Say that a vertex $v \in C_i$ is $\emph{good}$ if 
    \begin{enumerate}[label=(\textbf{\arabic*})]
    \item $\frac{d(k-1)}{k}\left( 1 - \sqrt{\frac{\lambda}{d}}\right) \leq \deg(v) \leq \frac{d(k-1)}{k}\left( 1 + \sqrt{\frac{\lambda}{d}}\right)$, and \label{prop_goodbad_degree}
    \item $\| f_v - \mu_i\|^2_2 \leq \ \frac{0.05k}{n} $, and \label{prop_goodbad_own}
    \item for all $j \neq i$, $\|f_v - \mu_j\|^2_2 \geq \frac{k}{n} $.\label{prop_goodbad_other}
    \end{enumerate}
    Say that the vertex is bad otherwise. Let  $B \subseteq V$ denote the set of bad vertices. 
\end{definition}

We will show that the oracle is correct on all good vertices. 
Observe that the set \(B\) of bad vertices is deterministic once the input graph and planted coloring are fixed. It consists of vertices whose degree or spectral embedding is atypical relative to their planted color class. 

\begin{restatable}{lemma}{Bbound}\label{lemma:B_bound}
Let $B$ be the set of bad vertices (as per \Cref{def:goodbad}). Then 
$$|B| \leq O\left( \sqrt{\frac{\lambda }{d}}\cdot n \right).$$
\end{restatable}
A similar bound for vertices with bad spectral properties was proved in \cite{DF16}. While their bound is similar, the definition of “bad” vertices in \cite{DF16} differs slightly from ours. The proof of \Cref{lemma:B_bound} is included in Appendix~\ref{sec:good-vertices}.

Additionally, we will use the fact that the $\ell_2$-norm of the embeddings of good vertices is not too large. 
\begin{restatable}{lemma}{goodnorm}\label{lemma:goodnorm}
        If $v$ is good (as per \Cref{def:goodbad}), then 
    $$\|f_v\|^2_2  =  O  \left( \frac{1}{n}\right).$$
\end{restatable}
The proof of \Cref{lemma:goodnorm} is presented in \Cref{sec:goodnorm}.  

\subsection{Analysis of \SDP{} (Proof of \Cref{lemma:SDP})}\label{sec:SDP}
In this section, we prove the correctness guarantee of \SDP{}, namely \Cref{lemma:SDP}. 
 Recall that we wish to apply the polynomial from \Cref{cor:polynomial} to the normalized adjacency matrix $\bar A$ in order to approximate the projection $\U\U^\top$. However, since \Cref{cor:polynomial} does not provide any guarantees for $p(1)$, we need to shift the top eigenvalue. Therefore, it will be convenient for the analysis to consider the matrix $\bar A - e_1e_1^\top$.
 \begin{definition}\label{def:M}
Let 
$$\Astar = \bar A - e_1e_1^\top. $$
\end{definition}
We start by showing that $\1_u^\top p(\Astar)^2\1_v$ approximates the inner product $\langle f_u, f_v\rangle.$ Later, in \Cref{claim:rw_to_Astar}, we will show that the quantity $\sqrt{\deg(u)}\sqrt{\deg(v)}\sum_{t,t'} c_tc_{t'} \left(\left \langle D^{-1/2}  \hat{p}^t_u , D^{-1/2} \hat{p}^{t'}_v\right \rangle - \frac{1}{2m}\right)$ computed by \Cref{alg:SDP} approximates $\1_u^\top p(\Astar)^2\1_v $. 
\begin{lemma}\label{lemma:poly_to_emnedding} Let $\delta>0$ be a sufficiently small constant, let $p$ be the polynomial from  $\Cref{cor:polynomial}$ (applied with the parameter $\delta$), and let $B$ be the set of bad vertices from \Cref{def:goodbad}. For every $u,v \in V \setminus B$, 
    $$ \left| \1_u^\top p(\Astar)^2\1_v  - \langle f_u, f_v\rangle \right| \leq \frac{1}{160}\frac{k}{n}.$$
\end{lemma}
\begin{proof}
The high-level idea is to decompose the matrix $\Astar$ according to its spectrum. We separate it into two parts: one corresponding to the bottom $k-1$ eigenvalues,
and one corresponding to the remaining $n-k+1$ eigenvalues.
By \Cref{lemma:C4}, the bottom $k-1$ eigenvalues are all close to $-1/(k-1)$.
Applying the polynomial~$p$ to this part therefore ``boosts'' these eigenvalues
up to (approximately) $1$, so this part closely approximates the spectral
projection $U_{[k-1]}U_{[k-1]}^\top$.
On the other hand, the remaining $n-k+1$ eigenvalues have  absolute value at most
$O(\sqrt{\lambda/d})$ (again by \Cref{lemma:C4}).
The polynomial~$p$ maps these values close to zero, so the error introduced by this part is small.

We now do this more formally. Let $\Astar = U \Sigma^{\star}U^\top$ denote the eigendecomposition of $\Astar$, where
\begin{equation*}\label{eq:sigma}
   \Sigma^{\star}= \diag\!\left(0, \lambda_2, \lambda_3, \ldots, \lambda_n\right).
\end{equation*}
Observe that $\bar A$ and $\Astar$ share the same eigenbasis, so the matrix $U_{[k-1]}$ from \Cref{def:spec-emb} is the matrix consisting of the last $k-1$ columns of $U$. 

Write $\Sigma^{\star}_{[k-1]}$ for the $(k-1)\times(k-1)$ diagonal matrix consisting of the bottom $k-1$ diagonal entries of $\Sigma^{\star}$, and $\Sigma^{\star}_{\mathrm{top}}$ for the $(n-k+1)\times(n-k+1)$ diagonal matrix consisting of the remaining (top) entries. Let $U_{\mathrm{top}}$ denote the matrix formed by the first $n-k+1$ columns of $U$.

With this notation, we can decompose
\begin{equation}\label{id:Sigma}
p(\Astar)^2
= U p(\Sigma^{\star})^2 U^\top
= U_{[k-1]} p(\Sigma^{\star}_{[k-1]})^2 U_{[k-1]}^\top
\;+\;
U_{\mathrm{top}} p(\Sigma^{\star}_{\mathrm{top}})^2 U_{\mathrm{top}}^\top.
\end{equation}
Recalling from \Cref{def:spec-emb} that 
$ \langle f_u, f_v \rangle = \1_u^\top  \U \U^\top \1_v$, we can bound the error by decomposing $p(\Astar)^2$ into the two terms in \Cref{id:Sigma}:  
\begin{equation}\label{eq:rw_to_embedding_polynomial}
 \begin{aligned}
     \left| \1_u^\top p(\Astar)^2\1_v  - \langle f_u, f_v\rangle \right|  & = \left|\1_u^\top Up(\Sigma^{\star})^2U^\top \1_v - \1^\top _u\U \U^\top \1_v\right| \\
     & = \left| \1_u^\top U_{[k-1]} p(\Sigma^{\star}_{[k-1]})^2 U_{[k-1]}^\top \1_v
+
\1_u^\top U_{\mathrm{top}} p(\Sigma^{\star}_{\mathrm{top}})^2 U_{\mathrm{top}}^\top \1_v- \1^\top _u\U \U^\top \1_v \right|\\
    & =  \left|\1_u^\top\U\left(p(\Sigma^{\star}_{[k-1]})^2 - I_{k-1}\right)\U^\top\1_v + \1_u^\top  U_{\text{top}}p(\Sigma^{\star}_{\text{top}})^2 U_{\text{top}}^\top  \1_v\right|  \\
& \leq \left|f_u^\top \left(p(\Sigma^{\star}_{[k-1]})^2 - I_{k-1}\right)\ f_v\right| + \left|\1^\top _u U_{\text{top}}p\left(\Sigma^{\star}_{\text{top}}\right)^2 U_{\text{top}}^\top \1_v\right|. 
\end{aligned}
\end{equation}

Here the second equality follows by \Cref{id:Sigma}, and the last inequality follows by the triangle inequality and the identity $f_u = U_{[k-1]}^\top \1_u$. 

We now bound the first summand in \Cref{eq:rw_to_embedding_polynomial}. By \Cref{lemma:C4} and the definition of $\Sigma^{\star}_{[k-1]}$, the matrix  $\Sigma^{\star}_{[k-1]}$ is a diagonal matrix with entries in $\left[ -\frac{1}{k-1}- O\left( \sqrt{\frac\lambda d}\right), -\frac{1}{k-1}+O\left( \sqrt{\frac\lambda d}\right)\right]$. Therefore, by \Cref{cor:polynomial}, $p(\Sigma^{\star}_{[k-1]})^2$ is a diagonal matrix with entries in $[ 1- 3\delta, 1+3\delta]$. Hence, 
$$\left \| p(\Sigma^{\star}_{[k-1]})^2 - I_{k-1}\right\|_2 \leq 3\delta, $$ so

\begin{equation}\label{eq:term1}
     \left|f_u^\top (p(\Sigma^{\star}_{[k-1]})^2 - I_{k-1})\ f_v\right| \leq \left \| f_u \right \|_2 \left \|\left( p(\Sigma^{\star}_{[k-1]})^2 - I_{k-1}\right) f_v\right \|_2 \leq 3\delta \|f_u\|_2 \|f_v\|_2 \leq \frac{1}{320n}, 
\end{equation}
where the last inequality follows by \Cref{lemma:goodnorm} and the assumption that $\delta$ is a sufficiently small constant. 

We now bound the second summand in \Cref{eq:rw_to_embedding_polynomial}.
The matrix $U_{\text{top}}p(\Sigma^{\star}_{\text{top}})^2U_{\text{top}}^\top$ is positive semidefinite, and therefore  
 \begin{align*}
 \left|\1^\top _uU_{\text{top}}p(\Sigma^{\star}_{\text{top}})^2U_{\text{top}}^\top   \ \1_v\right|   & \leq \1^\top _uU_{\text{top}}p(\Sigma^{\star}_{\text{top}})^2U_{\text{top}}^\top  \1_u + \1_v^\top U_{\text{top}}p(\Sigma^{\star}_{\text{top}})^2U_{\text{top}}^\top    \1_v \\
 & \leq \left\| p\left( \Sigma^{\star}_{\mathrm{top}} \right)^2 \right\|_2 \left \|U_{\text{top}}^\top  \1_u\right\|^2_2  + \left\| p\left( \Sigma^{\star}_{\mathrm{top}} \right)^2 \right\|_2 \left \|U_{\text{top}}^\top  \1_v\right\|^2_2  \\
 & \leq 2 \left\| p\left( \Sigma^{\star}_{\mathrm{top}} \right)^2 \right\|_2, 
 \end{align*}
where the last inequality uses that $\left\| U_{\text{top}}^\top  \1_w\right\|^2_2 \leq \|U^\top \1_w\|^2_2 = \|\1_w\|^2_2 = 1$  for all $w \in V$.

By \Cref{lemma:C4} and the definition of $\Sigma^{\star}_{\text{top}}$, the matrix  $\Sigma^{\star}_{\text{top}}$ is a diagonal matrix with entries bounded in absolute value by $O(\sqrt{\lambda/d})$.   Therefore, by \Cref{cor:polynomial}, 
$$ \left\|p\left( \Sigma^{\star}_{\text{top}}\right)^2 \right\|_2 \leq n^{-2}.$$
So the second summand in \Cref{eq:rw_to_embedding_polynomial} can be bounded as 
 \begin{equation}\label{eq:term2}
 \begin{aligned}
  \left|\1^\top _uU_{\text{top}}p(\Sigma^{\star}_{\text{top}})^2U_{\text{top}}^\top   \ \1_v\right|   \leq 2   \left\| p\left( \Sigma^{\star}_{\mathrm{top}} \right)^2 \right\|_2\leq 2n^{-2} \leq \frac{1}{320n}. 
\end{aligned}
 \end{equation} 
 Combining \Cref{eq:rw_to_embedding_polynomial}, \Cref{eq:term1} and \Cref{eq:term2} gives the claim. 
 \end{proof}

The following lemma shows that the quantity $\sqrt{\deg(u)}\sqrt{\deg(v)}\sum_{t,t'} c_tc_{t'} \left(\left \langle D^{-1/2}  \hat{p}^t_u , D^{-1/2} \hat{p}^{t'}_v\right \rangle - \frac{1}{2m}\right)$  computed by \Cref{alg:SDP} has expected value $\1_u^\top p(\Astar)^2\1_v = \sum_{t,t'}c_tc_{t'} \langle \Astar^t\1_u, \Astar ^{t'}\1_v\rangle$.

\begin{lemma}\label{lemma:expected_val}
For every $u,v \in V$ and every $t,t'\geq 1$, 
$$\sqrt{\deg(u)} \sqrt{\deg(v)}\left( \langle D^{-1/2} M^t \1_u,  D^{-1/2}M^{t'}\1_v \rangle  - \frac{1}{2m}\right) =\langle \Astar ^t  \1_u, \Astar^{t'} \1_v \rangle.$$
\end{lemma}
\begin{proof}
We start by relating the left-hand side to $\langle \bar A^t  \1_u, \bar A^{t'} \1_v \rangle$. From the definition of $M = AD^{-1}$ and $\bar A = D^{-1/2}AD^{-1/2}$, for every $t$, it holds that $M^t = D^{1/2}\bar A^t D^{-1/2}$. Therefore, 
\begin{equation}\label{eq:M_to_A}
    \begin{aligned}
     \langle D^{-1/2}M^t \1_u,  D^{-1/2}M^{t'}\1_v \rangle & = \1_u^\top \left(M^t\right)^{\top}D^{-1}M^{t'}\1_v \\
     & = \1_u^\top  D^{-1/2}  (\bar A^t)^\top D^{1/2}D^{-1}D^{1/2}A^{t'}D^{-1/2}\1_v \\
     & = \left( D^{-1/2}\1_u\right)^{\top}   (\bar A^t)^\top A^{t'}D^{-1/2}\1_v \\
     & = \frac{\langle \bar A^t  \1_u, \bar A^{t'} \1_v \rangle }{\sqrt{\deg(u)}\sqrt{\deg(v)} },
\end{aligned}
\end{equation}
\noindent
Next, we relate this to $\langle \Astar ^t  \1_u, \Astar^{t'} \1_v \rangle$. Recall from \Cref{obs:e1} that $e_1 =  \frac{D^{1/2}\1}{\sqrt{2m}}$. Since $\bar A e_1 = e_1$, we get
$$\Astar^t = (\bar A-e_1e_1^\top )^t = \bar A ^t -e_1e_1^\top \qquad \text{for every $t\geq 1$}.$$ 
Therefore, 
\begin{equation}\label{eq:A_to_Astar}
\begin{aligned}
    \left \langle \Astar^{t} \1_u, \Astar^{t'} \1_v \right \rangle & =   \left\langle \bar A^{t} \1_u, \bar A^{t'} \1_v \right \rangle - \langle e_1, \1_u\rangle \langle e_1, \1_v\rangle  \\
   & =  \left \langle \bar A ^{t} \1_u, \bar A^{t'} \1_v\right \rangle  - \frac{\langle \1_u ,  D^{1/2}\1\rangle \langle \1_v ,  D^{1/2}\1\rangle }{2m}\\ 
    & = \left \langle \bar A ^{t} \1_u, \bar A^{t'} \1_v\right \rangle - \frac{\sqrt{\deg(u)} \sqrt{\deg(v)}}{2m}.
\end{aligned}
\end{equation}
Combining \Cref{eq:M_to_A} and \Cref{eq:A_to_Astar} gives the lemma. 
\end{proof}

To prove \Cref{lemma:SDP}, we need to argue that we can estimate the quantity \newline 
$\sqrt{\deg(u)}\sqrt{\deg(v)}\sum_{t,t'} c_tc_{t'} \left(\left \langle D^{-1/2}M^t \1_u , D^{-1/2} M^{t'}\1_v \right \rangle - \frac{1}{2m}\right)$ to a sufficiently high precision. We will argue this on a term-by-term basis, using the collision counting result stated in \Cref{lemma:col-count}.

\begin{restatable}[Weighted Collision Counting]{lemma}{colcount}\label{lemma:col-count}
Let $p,q$ be two distributions $p,q \in \mathbb{R}^n$,  let $\xi \in (0,1)$ be the desired precision and let $\zeta \in (0,1)$ be the desired failure probability. Let $w_1, \dots w_n \in (0,w_{\max}]$ and let $W = \diag(w_1, \dots w_n)$. One can estimate $p^\top Wq$ up to precision $\xi \cdot w_{\max}(\|p\|^2_2 +\|q\|^2_2)$ with success probability $1-\zeta$ using $r =  5\xi^{-2}\zeta^{-1}(\|p\|_2^2  +\|q\|_2^2 )^{-1/2}$ samples from each of $p$ and $q$. That is, it is possible to compute $\hat{z} \in \mathbb{R}$ such that 
    $$ \Pr\left[ | p^\top W q - \hat{z}| > \xi\cdot w_{\max} \left(\|p\|^2_2 + \|q\|^2_2\right)\right]  \leq \zeta.$$
\end{restatable}
 The proof of \Cref{lemma:col-count} is a standard collision counting argument, and is included in Appendix \ref{sec:col-count}. 

We will apply \Cref{lemma:col-count} to the distributions $p = M^t \1_u$ and $q = M^{t'}\1_v$, and the weight matrix $W = D^{-1}$. Since both the error and the sample complexity in \Cref{lemma:col-count} depend on $\| p\|_2 = \|M^t \1_u \|_2$,  we need to bound the norm $\|M^t \1_u\|_2$ for $u \in V \setminus B$. This is achieved by \Cref{lemma:Mnorm} below. 

 \begin{lemma}\label{lemma:Mnorm}
  Let $C$ be a sufficiently large constant, and let  
 $t \geq \frac{C \log n}{\log(d/\lambda) }$. Then for every $u \in V \setminus B$, 
   $$\| \bar A ^t \1_u\|^2_2 = \Theta\left( \frac{1}{n}\right) \qquad \text{and} \qquad  \| M^t \1_u\|^2_2 = \Theta\left( \frac{1}{n}\right). $$
\end{lemma}

\begin{proof}
	Recall that  $1 = \lambda_1 \geq  \dots  \geq \lambda_n$ denotes the eigenvalues of $\bar A$, and $e_1, \dots e_n$ denotes the corresponding eigenvectors.
	We have 
		\begin{equation}\label{eq:l2_expression}
		\| \bar A ^t \1_u\|^2_2   = \sum_{i=1}^{n}\lambda_i^{2t} \langle \1_u, e_i \rangle ^2  =  \frac{\langle \1_u, D^{1/2}\1 \rangle ^2}{2m} + \sum_{i=2}^{n}\lambda_i^{2t}\langle \1_u, e_i \rangle ^2 =  \frac{\deg(u)}{2m} + \sum_{i =2}^{n}\lambda_i^{2t} \langle \1_u, e_i \rangle ^2 
	\end{equation}
	where the second transition uses $\lambda_1 = 1$ and $e_1  = \frac{D^{1/2}\1}{\sqrt{2m}} $, and the last transition uses $ \langle\1_u,  D^{1/2}\1\rangle ^2 = \deg(u)$. 
	By \Cref{cor:m}, $m = O(nd)$ and by the model assumption (\Cref{def:model}), $\deg(u) = \Omega(d)$. 
	Since every term in \Cref{eq:l2_expression} is non-negative, this gives $  \|\bar A ^t \1_u\|^2_2 \geq \frac{\deg(u)}{2m}  = \Omega\left(\frac{1}{n}\right)$. 
	
	We now prove that $  \|\bar A ^t \1_u\|^2_2  = O(1/n).$  By \Cref{lemma:C4}, for every $i \in [2,n-k+1]$, 
	$$|\lambda_i|\leq C'\sqrt{\frac{\lambda}{d}},$$
for some constant $C'$. 
By the assumption that $C$ is sufficiently large, we can assume that $t \geq \frac{C \log n}{\log(d/\lambda)} \geq \frac{\log n}{\log(\sqrt{d/\lambda}/C')}$. This gives
    $$ \lambda_i^{2t} \leq  \left(C' \sqrt{\frac{\lambda}{ d}}\right)^{2t} \leq   n^{-2}.$$
Substituting this into \Cref{eq:l2_expression} gives
	\begin{align*}
		\| \bar A ^t \1_u\|^2_2 & = \frac{\deg(u)}{2m} + \sum_{i =2}^{n-k+1}\lambda_i^{2t} \langle \1_u, e_i \rangle ^2 +\sum_{i =n-k+2}^{n}\lambda_i^{2t} \langle \1_u, e_i \rangle ^2 \\
		& \leq \frac{\deg(u)}{2m} + \frac{1}{n^2} \sum_{i = 2}^{n-2}\langle \1_u, e_i \rangle ^2 + \sum_{i =n-k+2}^{n}\lambda_i^{2t} \langle \1_u, e_i \rangle ^2\\
		& \leq \frac{\deg(u)}{2m} + \frac{1}{n^2} \|\1_u\|^2_2 + \|f_u\|^2_2 \\
		&= O\left(  \frac{1}{n}\right).
	\end{align*}
        Here, the last transition uses that $\|f_u\|_2^2 = O(1/n)$ by \Cref{lemma:goodnorm}, since $u \in V \setminus B$, and that $\deg(u) \leq d$ and $m = \Omega(nd)$, by \Cref{cor:m}.

       The guarantee on $\|M^t \1_u\|_2^2$ follows from $M^t \1_u = D^{1/2}\bar A ^t D^{-1/2}\1_u$, together with the assumption that $\deg(v) = \Theta(d)$ for all $v \in V$ (see \Cref{cor:m}). 
\end{proof}
We are now ready to prove the correctness of \SDP{}  (\Cref{lemma:SDP}). 
\begin{proof}[Proof of \Cref{lemma:SDP}] 
Let $C$ be a sufficiently large constant and let $\delta > 0$ be a sufficiently small constant. Let $$t_{\min} \coloneqq \frac{10C\log n}{\log(d/\lambda)}, $$
and let
$$ p(x) = x^{t_{\min}}q(x)$$
be the polynomial from \Cref{cor:polynomial} applied with parameters $\delta, C$ and $t = t_{\min}$. Let
 $$t_{\Delta} =  \deg q = O\left(\sqrt{\lambda/d} \log n\right).$$ 
 For a sufficiently large constant $\kappa$, let 
 \begin{equation}\label{eq:r}
     r= \kappa\sqrt{n}\cdot (\max |\coeff(p)| )^4(t_{\Delta}+1)^6. 
 \end{equation}
\noindent 
By \Cref{cor:polynomial}, $\max |\coeff(p)| =n^{O\left( \frac{1}{\log(d/\lambda)}\right)} $, so 

$$r =  \widetilde O\left( \sqrt{n} \cdot n^{O\left(\frac{1}{\log(d/\lambda)}\right)}\right).$$

We now show that with probability at least $0.9$, the quantity computed by the algorithm approximates $\sum_{t,t'}c_t, c_{t'}\langle \Astar \1_u , \Astar \1_v\rangle = \langle p(\Astar)\1_u , p(\Astar)\1_v\rangle $ to a high precision. 

\begin{claim}\label{claim:rw_to_Astar}
    $$ \Pr\left[\left| \sqrt{\deg(u)}\sqrt{\deg(v)}\sum_{t,t'} c_tc_{t'} \left(\left \langle D^{-1/2}\hat{p}^t_u , D^{-1/2} \hat{p}^{t'}_v\right \rangle - \frac{1}{2m}\right) - \langle p(\Astar)\1_u , p(\Astar)\1_v\rangle \right| \leq \frac{1}{160}\frac{k}{n}\right] \geq 0.9.$$
\end{claim}
\begin{proof} Fix $t, t' \in [t_{\min}, t_{\min} + t_{\Delta}]$. Set  $p = M^t \1_u$, $q = M^{t'}\1_v$ and $W =  D^{-1}$ in \Cref{lemma:col-count}.
Then 
$$p^\top Wq =\langle  D^{-1/2} M^t \1_u, D^{-1/2}M^{t'}\1_v \rangle $$ 
and 
    $$\| p\|^2_2 + \| q\|^2_2  = \|  M^t \1_u\|^2_2 + \|   M^{t'} \1_v\|^2_2 = \Theta\left( \frac{1}{n}\right),$$
where the last equality follows by \Cref{lemma:Mnorm}. 
Therefore, applying \Cref{lemma:col-count} with $w_{\max} = O(1/d)$ , $\zeta = \frac{0.1}{(t_{\Delta}+1)^2}$, $\xi = \frac{\gamma}{|c_t||c_{t'}|(t_{\Delta}+1)^2}$ for a sufficiently small constant $\gamma$, and $r$ as in \Cref{eq:r}, gives
\begin{equation*}\label{eq:rw_to_M}
    \Pr\left[\left| \left\langle D^{-1/2}\hat{p}_u^{t}, D^{-1/2}\hat{p}_v^{t'}\right \rangle  - \langle  D^{-1/2} M^t \1_u, D^{-1/2}M^{t'}\1_v \rangle \right| \leq \frac{1}{160 |c_{t}| |c_{t'}|(t_{\Delta}+1)^2} \cdot \frac{k}{d n} \right] \geq 1- \frac{0.1}{(t_{\Delta}+1)^2}.
\end{equation*}
\noindent 
Taking a union bound over the $(t_{\Delta} + 1)^2$ possible pairs of $t, t' \in [t_{\min}, t_{\min} + t_{\Delta}]$ gives
 \begin{equation}\label{eq:rw_to_embedding_term1}
 \begin{aligned}
     \Pr \left[ \left| \sum_{t,t'} c_t c_{t'}\left( \left \langle D^{-1/2}\hat{p}^t_u,    D^{-1/2}\hat{p}^{t'}_v \right \rangle - \langle  D^{-1/2} M^t \1_u, D^{-1/2}M^{t'}\1_v \rangle \right)\right| \leq \frac{ k}{160dn} \right] \geq 0.9. 
\end{aligned}
 \end{equation}
 From \Cref{lemma:expected_val}, we have 
\begin{align*}
& \sum_{t,t'} c_t c_{t'}\left( \left \langle D^{-1/2}\hat{p}^t_u,    D^{-1/2}\hat{p}^{t'}_v \right \rangle - \langle  D^{-1/2} M^t \1_u, D^{-1/2}M^{t'}\1_v \rangle \right) \\
 =& \sum_{t,t'} c_t c_{t'}\left( \left \langle D^{-1/2}\hat{p}^t_u,    D^{-1/2}\hat{p}^{t'}_v \right \rangle - \frac{1}{2m} \right) - \sum_{t,t'}c_tc_{t'}\left(\langle  D^{-1/2} M^t \1_u, D^{-1/2}M^{t'}\1_v \rangle- \frac{1}{2m} \right) \\
  =& \sum_{t,t'} c_t c_{t'}\left( \left \langle D^{-1/2}\hat{p}^t_u,    D^{-1/2}\hat{p}^{t'}_v \right \rangle - \frac{1}{2m} \right) - \frac{1}{\sqrt{\deg(u)}\sqrt{\deg(v)}}\sum_{t,t'}c_t c_{t'} \langle  \Astar^t \1_u, \Astar^{t'}\1_v\rangle \\
  =&  \sum_{t,t'} c_t c_{t'}\left( \left \langle D^{-1/2}\hat{p}^t_u,    D^{-1/2}\hat{p}^{t'}_v \right \rangle - \frac{1}{2m} \right) - \frac{1}{\sqrt{\deg(u)}\sqrt{\deg(v)}}\langle p(\Astar)\1_u, p(\Astar)\1_v\rangle.
\end{align*}
Combining with \Cref{eq:rw_to_embedding_term1} gives the claim. 
\end{proof}

Finally, by \Cref{lemma:poly_to_emnedding}, since $u,v \in V \setminus B$, we have 
\begin{equation}\label{eq:SDP1}\left| \langle p(\Astar)\1_u , p(\Astar)\1_v\rangle - \langle f_u, f_v\rangle  \right| \leq \frac{1}{160}\frac{k}{n}.
\end{equation}
Combining  \Cref{claim:rw_to_Astar} and \Cref{eq:SDP1} gives 
$$ \Pr\left[\left |\SDP(u,v) - \langle f_u, f_v\rangle \right| \leq \frac{1}{80}\frac{k}{n} \right]  \geq 0.9. $$

\paragraph{Running time.} 
The running time of $\SDP$ is dominated by the $2(t_{\Delta}+1) = \widetilde O(1)$ calls to $\randomwalks$ (\Cref{alg:randwalks}) in lines \ref{line:sdp_rwu}-\ref{line:sdp_rwv}, and the $(t_{\Delta}+1)^2 = \widetilde O(1)$ inner product computations in line \ref{line:sdp_return}.  Each call to $\randomwalks(r,t,\cdot)$ takes time $O(t \cdot r) = \widetilde O(r)$ and returns a vector $\hat{p}^t_u$ of support at most $r$. Then, each of the inner product computations $\langle  D^{-1/2} \hat{p}^t_u, D^{-1/2} \hat{p}^{t'}_v\rangle $ in line \ref{line:sdp_return} can be performed in time 
$$O\left (\min \left \{ |\supp(D^{-1/2}\hat{p}^{t}_u)|, |\supp(D^{-1/2}\hat{p}^{t'}_v)|\right \}\right) = O(r)$$  by computing $\sum_{w\in \supp(\hat{p}^{t}_u)  \cap \supp(\hat{p}^{t'}_v)} \hat{p}^{t}_u(w)\hat{p}^{t'}_v(w)/\deg(w)$.
So the overall running time is 
$\widetilde O(r) =\widetilde O\left(n^{1/2 + O\left(\frac{1}{\log(d/\lambda)}\right)}\right).$
\end{proof}
\subsection{Analysis of \approxdist{} (Proof of \Cref{lemma:approxdist} and \Cref{cor:approxdist_decision})}\label{sec:apprxdst}
In this section, we prove \Cref{lemma:approxdist} and \Cref{cor:approxdist_decision}. 
We start by proving \Cref{lemma:approxdist}. 

\begin{proof}[Proof of \Cref{lemma:approxdist}]
Let $i_{\max} = 5 \cdot 10^6 \log n$ denote the number of iterations in \approxdist{} (\Cref{alg:approxdist}). Let $u,v \in V \setminus B$.  For every $i \in [1, i_{\max}]$, each of the calls to $\SDP{}$ in lines~\ref{line:dotu},~\ref{line:dotv},~\ref{line:dotuv} succeeds with probability at least $0.9$. Therefore, by a union bound, with probability at least $0.7$, it holds that
\begin{align*}
    \left| \| f_u -  f_v\|_{\apx,i}^{2}- \| f_u -  f_v\|_2^2 \right| & =    \left| \left( \|f_u\|_{\apx,i}^{2} + \|f_v\|_{\apx,i}^{2}- 2 \langle f_u, f_v\rangle_{\apx,i}\right) - \left(  \|f_u\|_2^2 + \|f_v\|_2^2- 2 \langle f_u, f_v\rangle\right) \right|  \\
    & \leq \left| \|f_u\|_{\apx,i}^{2} - \|f_u\|^2_2\right|+\left| \|f_v\|_{\apx,i}^{2} - \|f_v\|^2_2\right|+2\left| \langle f_u, f_v\rangle _{\apx,i} - \langle f_u, f_v \rangle \right| \\
    & \leq \frac{k}{80n} + \frac{k}{80n} +  \frac{2k}{80n} \\
    & =  \frac{k}{20n},
\end{align*}
where the third transition holds by \Cref{lemma:SDP}. 
Therefore, by a Chernoff bound, for each pair $u,v \in V \setminus B$, 
$$ \Pr\left[ \left| \median_{i \in [1,i_{\max}]} (\|f_u - f_v\|_{\apx,i}^{2}) - \|f_u - f_v\|^2_2 \right| \geq  \frac{k}{20n} \right] \leq \exp( -0.00025 i_{\max}) \leq n^{-102}, $$
where the last inequality follows from the setting $i_{\max} =  5 \cdot 10^6\log n$. 
Taking a union bound over all possible pairs of $u,v \in V \setminus B$ (of which there are at most $n^2$) gives the required guarantee. 

\paragraph{Running time.} By \Cref{lemma:SDP}, each call to \SDP{} takes time  $\widetilde O\left(n^{1/2 + O\left(\frac{1}{\log(d/\lambda)}\right)}\right)$. Since the total number of calls to  \SDP{} is $O(\log n)  = \widetilde O(1)$, we get the claimed running time. 
\end{proof}

Next, we prove \Cref{cor:approxdist_decision}. 
\begin{proof}[Proof of \Cref{cor:approxdist_decision}]
Condition on the event that for all $u,v \in V \setminus B,$ 
\begin{equation}\label{eq:approx_guarantee}
    \left |\approxdist(u,v) - \| f_u- f_v\|_2^2 \right| \leq \frac{k}{20n} . 
\end{equation}
By \Cref{lemma:approxdist}, this happens with probability at least $1-n^{-100}$. 
Suppose that $\chi(u) = \chi(v) = i$. Then
\begin{equation}\label{eq:uv_in_same}
    \begin{aligned}
\|f_u - f_v\|_2 & \leq       \|f_u -\mu_i \|_2 +    \|f_v - \mu_i \|_2 && \text{by the triangle inequality} \\
& \leq 2 \cdot \sqrt{\frac{0.05k}{n}} && \text{by \Cref{def:goodbad}} \\
& = \sqrt{\frac{0.2k}{n}},
\end{aligned}
\end{equation}
which gives
\begin{align*}
    \approxdist(u,v) & \leq \left| \approxdist(u,v) - \|f_u - f_v\|^2_2 \right| + \|f_u - f_v\|_2^2 && \text{by the triangle inequality} \\
    & \leq \frac{0.05k}{n} +\frac{0.2k}{n} && \text{by \eqref{eq:approx_guarantee} and \eqref{eq:uv_in_same}} \\
    & = \frac{k}{4n}. 
\end{align*}
Suppose instead that $\chi(u)= i \neq \chi(v)$. 
Then 
\begin{equation}\label{eq:uv_in_diff}
\begin{aligned}
    \|f_u - f_v \|_2 & \geq   \|f_v - \mu_i\|_2 -    \|f_{u} - \mu_i\|_2 && \text{by the triangle inequality} \\
& \geq  \sqrt{\frac{k}{n}} - \sqrt{\frac{0.05k}{n}}&& \text{by \Cref{def:goodbad}} \\
& \geq \sqrt{\frac{0.56k}{n}}, 
\end{aligned}
\end{equation}
which gives 
\begin{align*}
    & \approxdist(u,v) \\
 \geq & \|f_u - f_v\|_2^2 - \left| \approxdist(u,v) - \|f_u - f_v\|^2_2 \right|   && \text{by the triangle inequality} \\
     \geq &  \frac{0.56k}{n} - \frac{0.05k}{n}  && \text{by \eqref{eq:approx_guarantee} and \eqref{eq:uv_in_diff}} \\
     >&   \frac{k}{2n}. 
\end{align*}
\end{proof}

\subsection{Construction of the polynomial (Proof of \Cref{cor:polynomial})}\label{sec:polynomial}
In this section, we prove \Cref{cor:polynomial}. More generally, we prove: 
\begin{thm}\label{thm:polynomial}
 Let $\alpha \in [-1,1]$ be a non-zero constant (independent of $n$),  and let $\delta >0$. For every  $\epsilon>0 $ smaller than a sufficiently small constant (depending on $|\alpha|$ and $\delta$), and every $\frac  {10 \log n}{\log(1/\epsilon)} \leq  t \leq O\left( \frac{\log n}{\log(1/\epsilon)}\right)$, there exists a polynomial~$p$ of the form $$p(x) = x^t q(x)$$ 
where $q$ is a polynomial of degree $O(\epsilon \log n)$ such that the coefficients of $p$ are bounded in absolute value by $ n^{O_{\alpha,\delta}\left(\frac{1}{\log(1/\epsilon)}\right)}$ and 
\begin{enumerate}[label=(\textbf{\arabic*})]
    \item $\left |p(x)-1\right| \leq \delta$ for $x \in [\alpha - \epsilon, \alpha +\epsilon]$ \label{con:poly1general}
    \item $|p(x)|\leq  n^{-2}$  for   $|x|\leq \epsilon$.  \label{con:poly0}
\end{enumerate}
\end{thm}
As a corollary of \Cref{thm:polynomial}, we obtain \Cref{cor:polynomial}. 
\begin{proof}[Proof of  \Cref{cor:polynomial}]
Immediate from \Cref{thm:polynomial} by setting $\alpha = \frac{-1}{(k-1)}$ and $\epsilon = C \sqrt{\lambda/d}$. 
\end{proof}
For the rest of the section, we prove \Cref{thm:polynomial}. 
Our polynomial builds on the polynomial construction from \cite{KKW}, and we use the following two results. 
\begin{thm}[Theorem A.2 in \cite{KKW}]\label{thm:Chebyshevapprox}
Let $t\geq 2$, $\epsilon \in  (0,0.1]$,  $d \geq 10 \epsilon t$, and let $$ Q(x) = \sum_{v = 0}^{d}T_v(x)c_{v,\epsilon,t}$$ be the polynomial obtained by taking the first $d+1$ terms of the Chebyshev expansion of $(1-\epsilon x)^{-t}$ on $[-1,1]$, where $T_v$ denotes the polynomial given by $T_v(\cos \theta) = \cos(v \theta)$. Then 
\begin{itemize}
    \item $\max_{ v\leq d} |c_{v,\epsilon,t}| \leq 2^{O(\epsilon t)}$, and 
    \item $\sup_{x \in [-1,1]} |Q(x) - (1-\epsilon x)^{-t}| =  O(d e^{-d})$. 
    \end{itemize}
\end{thm}

\begin{claim}[Claim A.3 in \cite{KKW}] \label{claim:coeff}There exists an absolute constant $C$ such that for every $\epsilon,t >0$ and every $d \geq 10\epsilon t$, the maximum coefficient of the polynomial $Q\left( \frac{1-x}{\epsilon}\right)$ is bounded in absolute value by $(1/\epsilon)^{C \cdot d}$, where $Q(y)$ denotes the degree $d$ polynomial from \Cref{thm:Chebyshevapprox}. 
\end{claim}

\begin{proof}[Proof of \Cref{thm:polynomial}]
The main idea is to split the proof into the two cases $\epsilon \lesssim |\alpha|/\log n$, in which case the polynomial $p(x) = (x/\alpha)^t$ is already enough, and the case $\epsilon \gtrsim |\alpha|/\log n$, in which case we will use the polynomial from \Cref{thm:Chebyshevapprox}. 
To keep the coefficients small enough in the first case, we choose
    $t \approx \frac{\log n}{\log(1/\epsilon)}$
rather than \(t \approx \log n\), as in \cite{KKW}. We then analyze the second
case more carefully for this choice of \(t\), instead of relying directly on the
analysis in \cite{KKW}.
    
First, we need to set up some parameters. 
Let $C$ denote the hidden universal constant in \Cref{thm:Chebyshevapprox} such that $\sup_{x \in [-1,1]} |Q(x) - (1-\epsilon x)^{-t}| \leq  C d e^{-d}$ for all $\epsilon$.
Suppose that $\epsilon$ satisfies
\begin{equation}\label{eq:eps}
\log(1/\epsilon) \geq \frac{\log(C/\delta)}{\delta}
\end{equation}
(this holds by the assumption that $\epsilon$ is smaller than a sufficiently small constant). Let $\tau\geq 10$ be the constant such that 
$$t = \frac{\tau \log n}{\log(1/\epsilon)}.$$
Additionally, define the constant
\begin{equation}\label{eq:kappa}
    \kappa \coloneqq \frac{\log(C/\delta)}{4\tau }\leq \frac{\log(C/\delta)}{40}. 
\end{equation}
\noindent
 We now split the proof into the two cases $\epsilon  \leq \frac{\kappa |\alpha|}{\log n}$ or $\epsilon > \frac{\kappa |\alpha|}{\log n}$. 
\paragraph{Case 1. } Suppose
$\epsilon \leq \frac{\kappa |\alpha|}{\log n}.$ Let 
$$p(x) = \frac{x^t}{\alpha^t}.$$
Then clearly $p$ is of the right form, and the only non-zero coefficient of $p$ is bounded in absolute value by

$$\left( \frac{1}{|\alpha|}\right)^{t} =  \left( \frac{1}{|\alpha|}\right)^{\frac{\tau \log n}{\log(1/\epsilon)}}  = n^{\frac{\tau \log(1/|\alpha|)}{\log(1/\epsilon)}} = n^{O_{\alpha}\left(\frac{1}{\log(1/\epsilon)}\right)},$$
by the assumption that $t =  \frac{\tau \log n}{\log(1/\epsilon)}$ for some constant $\tau$. 

We now verify Condition \ref{con:poly1general}. By the assumption $\epsilon \leq \frac{\kappa |\alpha|}{\log n}$,
\begin{equation}\label{eq:poly_case1}
\begin{aligned}
    \frac{4\epsilon t}{|\alpha|} &= \frac{4 \tau \epsilon\log n }{|\alpha |\log(1/\epsilon)}  && \text{by the definition of $\tau$}\\
   &  \leq \frac{4\tau \kappa}{\log(1/\epsilon)} && \text{by the assumption $\epsilon \leq \frac{\kappa |\alpha|}{\log n}$} \\
   & \leq \delta && \text{by \Cref{eq:kappa} and \Cref{eq:eps}}.
\end{aligned}
\end{equation}
 Therefore, 
\begin{align*}
   \max_{x \in [\alpha-\epsilon, \alpha + \epsilon]} \left|p(x) -1 \right| & =  \max_{x \in [\alpha-\epsilon, \alpha + \epsilon]} \left|\left( \frac{x}{\alpha} \right)^t-1 \right|  \\
   & = \max \left \{ \left( 1+\frac{\epsilon}{|\alpha|}\right)^t -1, 1-  \left( 1-\frac{\epsilon}{|\alpha|}\right)^t  \right\} \\
   & \leq \exp\left(\frac{2t\epsilon}{|\alpha|}\right)-1 \\
   & \leq \frac{4t\epsilon}{|\alpha|} && \text{since $e^x-1 \leq 2x$ for $x \in [0,1]$ } \\
   & \leq \delta && \text{by \Cref{eq:poly_case1}}. 
\end{align*}
Next, we verify Condition \ref{con:poly0}. We have 
\begin{align*}
   \max_{x \in [-\epsilon, \epsilon]} \left|p(x) \right|  &=  \max_{x \in[-\epsilon, \epsilon]} \left|\left( \frac{x}{\alpha} \right)^t \right|  = \left(\frac{\epsilon}{|\alpha|}\right)^t  \leq \exp\left( -t \log\left(\frac{|\alpha|}{\epsilon} \right)\right)  \\
   & \leq \exp\left( - \frac{-10\log(n)}{\log(1/\epsilon)} \log\left(\frac{|\alpha|}{\epsilon} \right)\right) \leq  n^{-2}, 
\end{align*}
where the fourth transition uses the assumption $t \geq \frac{10\log n}{\log(1/\epsilon)}$, and the last transition uses that $\epsilon$ is smaller than a sufficiently small constant depending on $\alpha$. 

\paragraph{Case 2.} Suppose $\epsilon \geq \frac{\kappa |\alpha|}{\log n}.$
 Let $Q(x)$ be the polynomial from \Cref{thm:Chebyshevapprox} applied with error parameter $\epsilon/|\alpha|$ and
 $$d =\frac{10 \log(1/\epsilon)\epsilon t}{|\alpha|} =   \frac{10\tau \epsilon \log n}{|\alpha|} = O_{\alpha}\left(\epsilon \log n \right).$$  Let
$$p(x) \coloneqq \left( \frac{x}{\alpha} \right)^t Q\left( \frac{1-\frac{x}{\alpha}}{\frac{\epsilon}{|\alpha|}}\right).$$
We start by bounding the coefficients of $p(x)$. By \Cref{claim:coeff} and the setting of $d$, the coefficients of the polynomial 
$$P(x) \coloneqq x^t Q\left( \frac{1-x}{\frac{\epsilon}{|\alpha|}}\right)$$
are bounded in absolute value by
$$ \left(\frac{|\alpha|}{\epsilon}\right)^{C_{\ref{claim:coeff}} \cdot d} =    n^{O_{\alpha}\left( \epsilon \log(1/\epsilon)\right)} = n^{O_{\alpha}(1/\log(1/\epsilon))},$$
where $C_{\ref{claim:coeff}}$ denotes the universal constant from \Cref{claim:coeff}. 
Since $p(x) = P(x/\alpha)$, the coefficients of $p$ are bounded in absolute value by 
\begin{align*}
    \max| \coeff(p)| & \leq \left(\frac{1}{|\alpha|}\right)^{\deg p} \cdot \max |\coeff(P)| \\
    & =  \left(\frac{1}{|\alpha|}\right)^{t + d}  n^{O_{\alpha}\left( \frac 1{\log(1/\epsilon)}\right)} \\
    &\leq \left(\frac{1}{|\alpha|}\right)^{\frac{\tau \log n}{\log(1/\epsilon)} +  \frac{10 \tau \epsilon\log n}{|\alpha|}}  n^{O_{\alpha}\left( \frac{1}{\log(1/\epsilon)}\right)} \\&= n^{O_{\alpha}\left(\frac{1}{\log(1/\epsilon)} \right)},
\end{align*}
as required. 
We now prove that Condition \ref{con:poly1general} holds. We have 
\begin{align*}
     \max_{x \in [\alpha-\epsilon, \alpha + \epsilon]} \left|p(x) -1 \right| & =   \max_{x \in [\alpha-\epsilon, \alpha + \epsilon]}  \left| \left( \frac{x}{\alpha}\right)^t   Q\left( \frac{1-x/\alpha}{\epsilon/|\alpha|}\right)-1\right|  \\
     & =    \max_{x \in [\alpha-\epsilon, \alpha + \epsilon]} \left|\left( \frac{x}{\alpha}\right)^t\right| \left|   Q\left( \frac{1-x/\alpha}{\epsilon/|\alpha|}\right)-\left( \frac{x}{\alpha}\right)^{-t}\right| \\
          & \leq  \left(1+\frac{\epsilon}{|\alpha|}\right)^t    \max_{x \in [\alpha-\epsilon, \alpha + \epsilon]}  \left|   Q\left( \frac{1-x/\alpha}{\epsilon/|\alpha|}\right)-\left( \frac{x}{\alpha}\right)^{-t}\right| \\
    & = \left(1+\frac{\epsilon}{|\alpha|}\right)^t \max_{z \in [-1,1]}\left|Q(z) - \left (1-\frac{\epsilon}{|\alpha|} z\right)^{-t} \right|&& \text{by setting $z = \frac{1-x/\alpha}{\epsilon/|\alpha|}$} \\
     & \leq C \cdot  \left(1+\frac{\epsilon}{|\alpha|}\right)^t de^{-d} && \text{by \Cref{thm:Chebyshevapprox} and setting of $C$} \\
    & \leq C \cdot  \exp \left(  \frac{\epsilon t}{|\alpha|}\right) \exp \left(-\frac{d}{2}\right)&& \text{since $x\leq e^{x/2}$ and $1+x \leq e^x$ for $x \in \R$.} \\
    \end{align*}
Therefore, 
\begin{align*}
\max_{x \in [\alpha-\epsilon, \alpha + \epsilon]} \left|p(x) -1 \right| 
    & =  C \cdot  \exp \left(  \frac{\epsilon t}{|\alpha|}-5\log(1/\epsilon)\frac{\epsilon}{|\alpha|} t\right) && \text{by the setting $d= 10 \log(1/\epsilon)\epsilon t/|\alpha|$}\\
    & \leq  C \cdot  \exp \left(  -4\log(1/\epsilon)\frac{\epsilon}{|\alpha|} t\right)\\
    & =   C \cdot  \exp \left(  -4\frac{\epsilon}{|\alpha|} \tau \log n \right) && \text{by the setting  $t =  \frac  {\tau \log n}{\log(1/\epsilon)}$}\\
     & \leq C\exp\left( -4 \kappa \tau \right) && \text{by the assumption $\epsilon \geq \frac{\kappa |\alpha|}{\log n}$} \\
     & = \delta && \text{by the setting of $\kappa$ \eqref{eq:kappa}. }
\end{align*}
\noindent
We now prove Condition \ref{con:poly0}. Since $p$ has at most $(d+1)$ non-zero coefficients, we have 
\begin{align*}
 \max_{x \in [-\epsilon, \epsilon]} \left|p(x) \right| & \leq \epsilon^{t} \cdot (d+1) \cdot \max|\text{coeff}(p)| \\ 
 & \leq \epsilon^{10\log n/\log(1/\epsilon)} \cdot O\left( \epsilon \log n/|\alpha|\right) \cdot n^{O_{\alpha}(1/\log(1/\epsilon))} \\
 & \leq n^{-10}\cdot n \cdot n^{O_{\alpha}(1/\log(1/\epsilon))} \\
& \leq n^{-2},
\end{align*}
where the last inequality uses that $\epsilon$ is smaller than a sufficiently small constant depending on $\alpha$, so that the value of the exponent $O_{\alpha}(1/\log(1/\epsilon))$ is at most $7$. 
\end{proof}

\subsection{Putting everything together (Proof of \Cref{thm:main-result})}\label{sec:proof_main}
We now combine  \Cref{cor:polynomial}, \Cref{lemma:approxdist}, \Cref{cor:approxdist_decision} and \Cref{lemma:B_bound} to obtain \Cref{thm:main-result}. 

\begin{proof}[Proof of \Cref{thm:main-result}]
The oracle runs the preprocessing algorithm \prepro{} (\Cref{alg:app_centers}) to obtain a set of $k$ color representatives $\{c_i\}_{i \in [k]}$. Then, for every query $u$, the oracle returns $\Color(u, \{c_i\}_{i \in [k]})$. 

 Let $B$ be the set of bad vertices from   \Cref{def:goodbad}. We now show that with probability at least $0.9$, the oracle is correct on all vertices in $V \setminus B$. 

  By \Cref{lemma:prepro_main}, with probability at least $0.95$, the output of \prepro{} is a valid set of color representatives (as per \Cref{def:rep}). \Cref{claim:main} below shows that conditioned on the success of \prepro{}, the oracle $\Color(v, \{c_i\}_{i\in [k]})$ returns the correct color (up to a permutation of colors) for every $v \in V \setminus B$: 
    \begin{claim}\label{claim:main}
        Suppose that $\{c_i\}_{i \in [k]}$ is a set of color representatives (as per \Cref{def:rep}). Let $\pi$ be the permutation such that $\pi(\chi(c_i)) = i$ for every $i \in [k]$.  Then with probability at least $1-n^{-100}$, for all $v \in V \setminus B$, 
 $$\Color(v, \{c_i\}_{i \in [k]}) = \pi(\chi(v)). $$
    \end{claim}
    \begin{proof}
    Let $v \in V\setminus B$, and let $c_i$ be the color representative such that $\chi(c_i) = \chi(v)$. Since  $\{c_j\}_{j \in [k]}$ is a set of color representatives, by \Cref{def:rep}, every $c_j$ belongs to a different color class, and $c_j \in V \setminus B$ for all $j \in [k]$. 
    Therefore, by \Cref{cor:approxdist_decision}, 
    $$\approxdist(v, c_i) \leq \frac{k}{4n}$$ and for all $j \neq i$, 
   $$ \approxdist(v, c_j) \geq \frac{k}{2n}.$$
   In particular, $\argmin_j \approxdist(v,c_j) = i$. By definition of the permutation $\pi$, we have $i=  \pi(\chi(c_i)) = \pi( \chi(v))$, so \Color{} returns $\pi( \chi(v))$ as claimed.
\end{proof}

Taking a union bound over the success of \prepro{} and \Cref{claim:main} gives that with probability at least $0.95-n^{-100} \geq 0.9$, every vertex in $V \setminus B$ is colored correctly (up to a permutation of the colors). In that case,  by \Cref{lemma:B_bound}, 
    $$| \{v : \Color(v, \{c_i\}_{i \in [k]}) \neq \pi(\chi(v))\}|  \leq |B| \leq O\left( \sqrt{\frac{\lambda }{d}}\cdot n\right). $$
   
\paragraph{Running time and space complexity.}
The preprocessing stage runs \prepro{} (\Cref{alg:app_centers}). By \Cref{lemma:prepro_main}, this has running time $\widetilde O\!\left(n^{1/2+O(1/\log(d/\lambda))}\right)$. In order to be consistent, the algorithm should store the $O(r) = \widetilde O\!\left(n^{1/2+O(1/\log(d/\lambda))}\right)$  random walks for each of the $k$ representative vertices, and use the same $O(r)$ random bits for each query. This can be done using space $\widetilde O\!\left(n^{1/2+O(1/\log(d/\lambda))}\right)$. 

Each query is answered by \Color{}  (\Cref{alg:color_main}), which calls \approxdist{} (\Cref{alg:approxdist}) $k=O(1)$ times. By \Cref{lemma:approxdist}, this takes time $\widetilde O\!\left(n^{1/2+O(1/\log(d/\lambda))}\right)$.
\end{proof}

\bibliographystyle{alpha}
\bibliography{references}

@inproceedings{DF16,
author = {David, Roee and Feige, Uriel},
title = {On the effect of randomness on planted 3-coloring models},
year = {2016},
isbn = {9781450341325},
publisher = {Association for Computing Machinery},
address = {New York, NY, USA},
url = {https://doi.org/10.1145/2897518.2897561},
doi = {10.1145/2897518.2897561},
booktitle = {Proceedings of the Forty-Eighth Annual ACM Symposium on Theory of Computing},
pages = {77–90},
numpages = {14},
location = {Cambridge, MA, USA},
series = {STOC '16}
}

@article{AK97,
  author       = {Noga Alon and
                  Nabil Kahal{\'{e}}},
  title        = {A Spectral Technique for Coloring Random 3-Colorable Graphs},
  journal      = {{SIAM} J. Comput.},
  volume       = {26},
  number       = {6},
  pages        = {1733--1748},
  year         = {1997},
  url          = {https://doi.org/10.1137/S0097539794270248},
  doi          = {10.1137/S0097539794270248},
  bibsource    = {dblp computer science bibliography, https://dblp.org}
}

@inproceedings{peng2023sublinear,
  title={Sublinear-time algorithms for {M}ax {C}ut, {M}ax {E}2{L}in(q), and unique label cover on expanders},
  author={Peng, Pan and Yoshida, Yuichi},
  booktitle={Proceedings of the 2023 Annual ACM-SIAM Symposium on Discrete Algorithms (SODA)},
  pages={4936--4965},
  year={2023},
  organization={SIAM}
}

@inproceedings{KKW,
author = {Michael Kapralov and Ekaterina Kochetkova and Weronika Wrzos{-}Kaminska},
title = {Spectral clustering in birthday paradox time},
booktitle = {Proceedings of the 2026 Annual ACM-SIAM Symposium on Discrete Algorithms (SODA)},
pages = {3057-3071},
doi = {10.1137/1.9781611978971.113},
URL = {https://epubs.siam.org/doi/abs/10.1137/1.9781611978971.113},
eprint = {https://epubs.siam.org/doi/pdf/10.1137/1.9781611978971.113}, 
  year={2026}
}

@article{SP23,
  title={A sublinear-time spectral clustering oracle with improved preprocessing time},
  author={Shen, Ranran and Peng, Pan},
  journal={Advances in Neural Information Processing Systems},
  volume={36},
  pages={41809--41828},
  year={2023}
}

@inproceedings{FKKLMW26,
author = {Hendrik Fichtenberger and Michael Kapralov and Ekaterina Kochetkova and Silvio Lattanzi and Davide Mazzali and Weronika Wrzos{-}Kaminska},
title = {Spectral Clustering with Side Information},
booktitle = {Proceedings of the 2026 Annual ACM-SIAM Symposium on Discrete Algorithms (SODA)},
chapter = {},
pages = {4327-4341},
doi = {10.1137/1.9781611978971.159},
URL = {https://epubs.siam.org/doi/abs/10.1137/1.9781611978971.159},
eprint = {https://epubs.siam.org/doi/pdf/10.1137/1.9781611978971.159},
  year={2026}
}

@article{BlumSpencer,
title = {Coloring Random and Semi-Random k-Colorable Graphs},
journal = {Journal of Algorithms},
volume = {19},
number = {2},
pages = {204-234},
year = {1995},
issn = {0196-6774},
doi = {https://doi.org/10.1006/jagm.1995.1034},
url = {https://www.sciencedirect.com/science/article/pii/S0196677485710346},
author = {A. Blum and J. Spencer}
}

@inproceedings{Karp72,
  author       = {Richard M. Karp},
  editor       = {Raymond E. Miller and
                  James W. Thatcher},
  title        = {Reducibility Among Combinatorial Problems},
  booktitle    = {Proceedings of a symposium on the Complexity of Computer Computations,
                  held March 20-22, 1972, at the {IBM} Thomas J. Watson Research Center,
                  Yorktown Heights, New York, {USA}},
  series       = {The {IBM} Research Symposia Series},
  pages        = {85--103},
  publisher    = {Plenum Press, New York},
  year         = {1972},
  url          = {https://doi.org/10.1007/978-1-4684-2001-2\_9},
  doi          = {10.1007/978-1-4684-2001-2\_9},
  bibsource    = {dblp computer science bibliography, https://dblp.org}
}

@article{NPcomplete,
  author       = {M. R. Garey and
                  David S. Johnson and
                  Larry J. Stockmeyer},
  title        = {Some Simplified {NP}-Complete Graph Problems},
  journal      = {Theor. Comput. Sci.},
  volume       = {1},
  number       = {3},
  pages        = {237--267},
  year         = {1976},
  url          = {https://doi.org/10.1016/0304-3975(76)90059-1},
  doi          = {10.1016/0304-3975(76)90059-1},
  bibsource    = {dblp computer science bibliography, https://dblp.org}
}

@inproceedings{Kumar17,
  author       = {Akash Kumar and
                  Anand Louis and
                  Madhur Tulsiani},
  editor       = {Satya V. Lokam and
                  R. Ramanujam},
  title        = {Finding Pseudorandom Colorings of Pseudorandom Graphs},
  booktitle    = {37th {IARCS} Annual Conference on Foundations of Software Technology
                  and Theoretical Computer Science, {FSTTCS} 2017, Kanpur, India, December
                  11-15, 2017},
  series       = {LIPIcs},
  volume       = {93},
  pages        = {37:1--37:12},
  publisher    = {Schloss Dagstuhl - Leibniz-Zentrum f{\"{u}}r Informatik},
  year         = {2017},
  url          = {https://doi.org/10.4230/LIPIcs.FSTTCS.2017.37},
  doi          = {10.4230/LIPICS.FSTTCS.2017.37},
  bibsource    = {dblp computer science bibliography, https://dblp.org}
}

@article{Buhai25,
  author       = {Rares{-}Darius Buhai and
                  Yiding Hua and
                  David Steurer and
                  Andor V{\'{a}}ri{-}Kakas},
  title        = {Finding Colorings in One-Sided Expanders},
  journal      = {CoRR},
  volume       = {abs/2508.02825},
  year         = {2025},
  url          = {https://doi.org/10.48550/arXiv.2508.02825},
  doi          = {10.48550/ARXIV.2508.02825},
  eprinttype    = {arXiv},
  eprint       = {2508.02825},
  bibsource    = {dblp computer science bibliography, https://dblp.org}
}

@article{Nilli,
author = {Nilli, A.},
title = {On the second eigenvalue of a graph},
year = {1991},
issue_date = {Aug. 28, 1991},
publisher = {Elsevier Science Publishers B. V.},
address = {NLD},
volume = {91},
number = {2},
issn = {0012-365X},
url = {https://doi.org/10.1016/0012-365X(91)90112-F},
doi = {10.1016/0012-365X(91)90112-F},
journal = {Discrete Math.},
month = aug,
pages = {207–210},
numpages = {4}
}

@article{FO05,
author = {Feige, Uriel and Ofek, Eran},
title = {Spectral techniques applied to sparse random graphs},
year = {2005},
issue_date = {September 2005},
publisher = {John Wiley \& Sons, Inc.},
address = {USA},
volume = {27},
number = {2},
issn = {1042-9832},
journal = {Random Struct. Algorithms},
month = sep,
pages = {251–275},
numpages = {25}
}

@article{FK81,
	author = {F{\"u}redi, Z. and Koml{\'o}s, J.},
	date = {1981/09/01},
	doi = {10.1007/BF02579329},
	id = {F{\"u}redi1981},
	isbn = {1439-6912},
	journal = {Combinatorica},
	number = {3},
	pages = {233--241},
	title = {The eigenvalues of random symmetric matrices},
	url = {https://doi.org/10.1007/BF02579329},
	volume = {1},
	year = {1981}}

@inproceedings{FKS89,
author = {Friedman, J. and Kahn, J. and Szemer\'{e}di, E.},
title = {On the second eigenvalue of random regular graphs},
year = {1989},
isbn = {0897913078},
publisher = {Association for Computing Machinery},
address = {New York, NY, USA},
url = {https://doi.org/10.1145/73007.73063},
doi = {10.1145/73007.73063},
booktitle = {Proceedings of the Twenty-First Annual ACM Symposium on Theory of Computing},
pages = {587–598},
numpages = {12},
location = {Seattle, Washington, USA},
series = {STOC '89}
}

@inproceedings{CzumajPS15,
  author       = {Artur Czumaj and
                  Pan Peng and
                  Christian Sohler},
  editor       = {Rocco A. Servedio and
                  Ronitt Rubinfeld},
  title        = {Testing Cluster Structure of Graphs},
  booktitle    = {Proceedings of the Forty-Seventh Annual {ACM} on Symposium on Theory
                  of Computing, {STOC} 2015, Portland, OR, USA, June 14-17, 2015},
  pages        = {723--732},
  publisher    = {{ACM}},
  year         = {2015},
  url          = {https://doi.org/10.1145/2746539.2746618},
  doi          = {10.1145/2746539.2746618},
  bibsource    = {dblp computer science bibliography, https://dblp.org}
}

@inproceedings{AG11,
  author       = {Sanjeev Arora and
                  Rong Ge},
  editor       = {Leslie Ann Goldberg and
                  Klaus Jansen and
                  R. Ravi and
                  Jos{\'{e}} D. P. Rolim},
  title        = {New Tools for Graph Coloring},
  booktitle    = {Approximation, Randomization, and Combinatorial Optimization. Algorithms
                  and Techniques - 14th International Workshop, {APPROX} 2011, and 15th
                  International Workshop, {RANDOM} 2011, Princeton, NJ, USA, August
                  17-19, 2011. Proceedings},
  series       = {Lecture Notes in Computer Science},
  volume       = {6845},
  pages        = {1--12},
  publisher    = {Springer},
  year         = {2011},
  url          = {https://doi.org/10.1007/978-3-642-22935-0\_1},
  doi          = {10.1007/978-3-642-22935-0\_1},
  bibsource    = {dblp computer science bibliography, https://dblp.org}
}

@inproceedings{Maxcut24,
  author       = {Agastya Vibhuti Jha and
                  Akash Kumar},
  editor       = {Karl Bringmann and
                  Martin Grohe and
                  Gabriele Puppis and
                  Ola Svensson},
  title        = {A Sublinear Time Tester for Max-Cut on Clusterable Graphs},
  booktitle    = {51st International Colloquium on Automata, Languages, and Programming,
                  {ICALP} 2024, Tallinn, Estonia, July 8-12, 2024},
  series       = {LIPIcs},
  volume       = {297},
  pages        = {91:1--91:17},
  publisher    = {Schloss Dagstuhl - Leibniz-Zentrum f{\"{u}}r Informatik},
  year         = {2024},
  url          = {https://doi.org/10.4230/LIPIcs.ICALP.2024.91},
  doi          = {10.4230/LIPICS.ICALP.2024.91},
  bibsource    = {dblp computer science bibliography, https://dblp.org}
}

@inproceedings{BHK25,
author = {Bafna, Mitali and Hsieh, Jun-Ting and Kothari, Pravesh K.},
title = {Rounding Large Independent Sets on Expanders},
year = {2025},
isbn = {9798400715105},
publisher = {Association for Computing Machinery},
address = {New York, NY, USA},
url = {https://doi.org/10.1145/3717823.3718137},
doi = {10.1145/3717823.3718137},
booktitle = {Proceedings of the 57th Annual ACM Symposium on Theory of Computing},
pages = {631–642},
numpages = {12},
location = {Prague, Czechia},
series = {STOC '25}
}

@inproceedings{Hsieh26,
author = {Jun-Ting Hsieh},
title = {Coloring 3-Colorable Graphs with Low Threshold Rank},
booktitle = {Proceedings of the 2026 Annual ACM-SIAM Symposium on Discrete Algorithms (SODA)},
chapter = {},
year = {2026}, 
pages = {718-726},
doi = {10.1137/1.9781611978971.27},
URL = {https://epubs.siam.org/doi/abs/10.1137/1.9781611978971.27},
eprint = {https://epubs.siam.org/doi/pdf/10.1137/1.9781611978971.27}
}

@article{Wigderson83,
  author       = {Avi Wigderson},
  title        = {Improving the Performance Guarantee for Approximate Graph Coloring},
  journal      = {J. {ACM}},
  volume       = {30},
  number       = {4},
  pages        = {729--735},
  year         = {1983},
  url          = {https://doi.org/10.1145/2157.2158},
  doi          = {10.1145/2157.2158},
  bibsource    = {dblp computer science bibliography, https://dblp.org}
}

@article{BK97,
  author       = {Avrim Blum and
                  David R. Karger},
  title        = {An {$\widetilde{O}(n^{3/14})$}-Coloring Algorithm for 3-Colorable Graphs},
  journal      = {Inf. Process. Lett.},
  volume       = {61},
  number       = {1},
  pages        = {49--53},
  year         = {1997},
  url          = {https://doi.org/10.1016/S0020-0190(96)00190-1},
  doi          = {10.1016/S0020-0190(96)00190-1},
  bibsource    = {dblp computer science bibliography, https://dblp.org}
}

@article{Blum94,
  author       = {Avrim Blum},
  title        = {New Approximation Algorithms for Graph Coloring},
  journal      = {J. {ACM}},
  volume       = {41},
  number       = {3},
  pages        = {470--516},
  year         = {1994},
  url          = {https://doi.org/10.1145/176584.176586},
  doi          = {10.1145/176584.176586},
  bibsource    = {dblp computer science bibliography, https://dblp.org}
}

@article{BR90,
  author       = {Bonnie Berger and
                  John Rompel},
  title        = {A Better Performance Guarantee for Approximate Graph Coloring},
  journal      = {Algorithmica},
  volume       = {5},
  number       = {3},
  pages        = {459--466},
  year         = {1990},
  url          = {https://doi.org/10.1007/BF01840398},
  doi          = {10.1007/BF01840398},
  bibsource    = {dblp computer science bibliography, https://dblp.org}
}

@article{KMS98,
  author       = {David R. Karger and
                  Rajeev Motwani and
                  Madhu Sudan},
  title        = {Approximate Graph Coloring by Semidefinite Programming},
  journal      = {J. {ACM}},
  volume       = {45},
  number       = {2},
  pages        = {246--265},
  year         = {1998},
  url          = {https://doi.org/10.1145/274787.274791},
  doi          = {10.1145/274787.274791},
  bibsource    = {dblp computer science bibliography, https://dblp.org}
}

@article{KLS00,
  author       = {Sanjeev Khanna and
                  Nathan Linial and
                  Shmuel Safra},
  title        = {On the Hardness of Approximating the Chromatic Number},
  journal      = {Comb.},
  volume       = {20},
  number       = {3},
  pages        = {393--415},
  year         = {2000},
  url          = {https://doi.org/10.1007/s004930070013},
  doi          = {10.1007/S004930070013},
  bibsource    = {dblp computer science bibliography, https://dblp.org}
}

@inproceedings{Khot01,
author = {Khot, S.},
title = {Improved Inapproximability Results for {M}ax{C}lique, Chromatic Number and Approximate Graph Coloring},
year = {2001},
isbn = {0769513905},
publisher = {IEEE Computer Society},
address = {USA},
booktitle = {Proceedings of the 42nd IEEE Symposium on Foundations of Computer Science},
pages = {600},
series = {FOCS '01}
}

@inproceedings{AC06,
  author       = {Sanjeev Arora and
                  Eden Chlamtac},
  editor       = {Jon M. Kleinberg},
  title        = {New approximation guarantee for chromatic number},
  booktitle    = {Proceedings of the 38th Annual {ACM} Symposium on Theory of Computing,
                  Seattle, WA, USA, May 21-23, 2006},
  pages        = {215--224},
  publisher    = {{ACM}},
  year         = {2006},
  url          = {https://doi.org/10.1145/1132516.1132548},
  doi          = {10.1145/1132516.1132548},
  bibsource    = {dblp computer science bibliography, https://dblp.org}
}

@inproceedings{DMR06,
author = {Dinur, Irit and Mossel, Elchanan and Regev, Oded},
title = {Conditional hardness for approximate coloring},
year = {2006},
isbn = {1595931341},
publisher = {Association for Computing Machinery},
address = {New York, NY, USA},
url = {https://doi.org/10.1145/1132516.1132567},
doi = {10.1145/1132516.1132567},
booktitle = {Proceedings of the Thirty-Eighth Annual ACM Symposium on Theory of Computing},
pages = {344–353},
numpages = {10},
location = {Seattle, WA, USA},
series = {STOC '06}
}

@article{KT17,
  author       = {Ken{-}ichi Kawarabayashi and
                  Mikkel Thorup},
  title        = {Coloring 3-Colorable Graphs with Less than \emph{n}\({}^{\mbox{1/5}}\)
                  Colors},
  journal      = {J. {ACM}},
  volume       = {64},
  number       = {1},
  pages        = {4:1--4:23},
  year         = {2017},
  url          = {https://doi.org/10.1145/3001582},
  doi          = {10.1145/3001582},
  bibsource    = {dblp computer science bibliography, https://dblp.org}
}

@inproceedings{KTY24,
  author       = {Ken{-}ichi Kawarabayashi and
                  Mikkel Thorup and
                  Hirotaka Yoneda},
  editor       = {Bojan Mohar and
                  Igor Shinkar and
                  Ryan O'Donnell},
  title        = {Better Coloring of 3-Colorable Graphs},
  booktitle    = {Proceedings of the 56th Annual {ACM} Symposium on Theory of Computing,
                  {STOC} 2024, Vancouver, BC, Canada, June 24-28, 2024},
  pages        = {331--339},
  publisher    = {{ACM}},
  year         = {2024},
  url          = {https://doi.org/10.1145/3618260.3649768},
  doi          = {10.1145/3618260.3649768},
  bibsource    = {dblp computer science bibliography, https://dblp.org}
}

@inproceedings{GKLMS21,
  author       = {Grzegorz Gluch and
                  Michael Kapralov and
                  Silvio Lattanzi and
                  Aida Mousavifar and
                  Christian Sohler},
  title        = {Spectral Clustering Oracles in Sublinear Time},
  booktitle    = {Proceedings of the 2021 {ACM-SIAM} Symposium on Discrete Algorithms,
                  {SODA} 2021, Virtual Conference, January 10 - 13, 2021},
  pages        = {1598--1617},
  publisher    = {{SIAM}},
  year         = {2021}
}

@inproceedings{CKK18,
  title={Testing Graph Clusterability: Algorithms and Lower Bounds},
  author={Chiplunkar, Ashish and Kapralov, Michael and Khanna, Sanjeev and Mousavifar, Aida and Peres, Yuval},
  booktitle={2018 IEEE 59th Annual Symposium on Foundations of Computer Science (FOCS)},
  pages={497--508},
  year={2018},
  organization={IEEE}
}
\appendix 
\section{Obtaining Color Representatives (Proof of \Cref{lemma:prepro_main})}\label{sec:prepro_analysis}
In this section, we prove \Cref{lemma:prepro_main}, which is restated below for the convenience of the reader. 

\preprolemma*
The analysis follows the analysis of the preprocessing routine in \cite{FKKLMW26}. 
\begin{lemma}\label{lemma:prepro_properties}
Let $S$ be a multiset of $10k\log k$ vertices sampled uniformly at random with replacement from~$V$. Then, 
\begin{enumerate}[label=(\textbf{\arabic*})]
    \item With probability at least $0.99$, $S$ does not contain any vertices from $B$, i.e.,  $\Pr[ S \cap B = \emptyset ] \geq 0.99$.
    \label{item:approxmu_event1}
    \item With probability at least $0.99$, $S$ contains at least one vertex from every color class, i.e., $\Pr[ \forall i, \ S \cap C_i \neq \emptyset ] \geq 0.99$.\label{item:approxmu_event2}
\end{enumerate}
\end{lemma}
\begin{proof}
We first show the statement \ref{item:approxmu_event1}, and then we show the statement \ref{item:approxmu_event2}. 
\paragraph{Statement \ref{item:approxmu_event1}.}  By \Cref{lemma:B_bound}, $ |B| \leq C\cdot  \sqrt{\frac{\lambda }{d}}n$ for some constant $C$. Since we assume  $k =O(1)$ and $\lambda/d$ is smaller than a sufficiently small constant (depending on $k$), this gives 
\begin{equation}\label{eq:prepro_B_bound}
    |B| \leq \frac{n}{10^3k\log k}.  
\end{equation} So
$$ \E_S[ |S \cap B|] = |S|\cdot \frac{|B|}{n} = \frac{10k\log k}{ 10^3 k\log k} =  10^{-2}. $$
 Therefore, with probability at least $0.99$ by Markov's inequality,  $S \cap B$ is empty.
\paragraph{Statement \ref{item:approxmu_event2}.} Fix $i$. The probability that a vertex sampled from $V$ uniformly at random is not from $C_i$, is $1 - \frac{|C_i|}{n} =  1 - \frac{1}{k}.$ Therefore, since $S$ is a multiset of $10 k\log k$ vertices sampled uniformly at random with replacement,
    \[ \Pr[S \cap C_i = \emptyset] =  \left(1 - \frac{1}{k}\right)^{|S|}  =  \left(1 - \frac{1}{k}\right)^{10k\log k} \leq \frac{1}{k^{10}}. \] 
\noindent
By a union bound, 
 \[ \Pr[ \ \exists i \ \text{ such that }\  S \cap C_i = \emptyset \ ]\leq  k\cdot \frac{1}{k^{10}} = \frac{1}{k^9} \leq 0.01,\]
 where the last inequality follows from $k \geq 2$. 
\end{proof}

\begin{definition}[Same-color graph]\label{def:same-color-graph} Given a set of vertices $S$, say that $H = (S, E_H)$ is the same-color graph induced by $S$ if the set of edges $E_H$ is exactly the set of pairs of vertices in $S$ belonging to the same color class i.e.,  
$$E_H = \{ (u,v) \in S\times S : \chi(u) = \chi(v)\}. $$
\end{definition}

We now prove \Cref{lemma:prepro_main}.
\begin{proof}[Proof of  \Cref{lemma:prepro_main}]
We first show in \Cref{claim:H_cliques} that conditioned on the event $S \cap B= \emptyset$, the graph $H$ computed after line \ref{line:end_init_H} in \Cref{alg:app_centers} is the same-color graph. 
\begin{claim}[$H$ is the same-color graph]\label{claim:H_cliques}
If $S \cap B = \emptyset$, then $H$ is the same-color graph induced by $S$ (as per \Cref{def:same-color-graph}) with probability at least $1-n^{-100}$.  
\end{claim}
\begin{proof}
Suppose that $S \cap B = \emptyset$. Then, every $u,v \in S$ belong to $V \setminus B$, so by \Cref{cor:approxdist_decision}, with probability at least $1-n^{-100}$, for every $u,v \in S$ that belong to the same color class, 
$$  \approxdist{(u,v)}  \leq  \frac{k}{4n} < \frac{k}{3n}, $$
and for every pair $u,v \in S$ that belongs to different color classes, 
$$  \approxdist{(u,v)}  \geq  \frac{k}{2n} > \frac{k}{3n}  $$

So the set of edges added to $H$ in line \ref{line:add_to_H} is exactly the set  $\{ (u,v) \in S \times S: \chi(u) = \chi(v)\}$, and so $H$ is the same-color graph induced by $S$. 
\end{proof}

By \Cref{claim:H_cliques}, $H$ is the same-color graph, and by property \ref{item:approxmu_event2} in \Cref{lemma:prepro_properties}, with probability at least $0.99$, it has exactly $k$ cliques. Every time line~\ref{line:select_center} in \prepro{} is triggered, we pick a vertex $u$, declare it a representative vertex by setting $c_i \coloneqq u$, and remove the rest of the vertices from $C_{\chi(u)}$ from $H$.
Hence, the algorithm stops after exactly $k$ iterations. 
Since the output is exactly $k$ vertices belonging to different color classes, and since $S \cap B = \emptyset$, the returned set is a set of color class representatives (as per \Cref{def:rep}). 

\paragraph{Running time.} The running time is dominated by the $|S|^2$ calls to \approxdist{} (\Cref{alg:approxdist}). By \Cref{lemma:approxdist}, each of these calls takes time $\widetilde O\left(n^{1/2 + O\left(\frac{1}{\log(d/\lambda)}\right)}\right)$. Since $|S| = 10k \log k  =O(1)$, the number of such calls is $O(1)$. 
\end{proof}

\section{Deferred Proofs from \Cref{sec:properties}}\label{sec:properties_pfs}
In this section, we prove the properties of the input instance stated in \Cref{sec:properties}, i.e., \Cref{lemma:C4}, \Cref{lemma:B_bound} and \Cref{lemma:goodnorm}. 

We start by proving an auxiliary lemma (\Cref{lemma:degrees}) that will be used in both the proof of \Cref{lemma:C4} and \Cref{lemma:B_bound}. It shows that almost all vertices have degree $\frac{d(k-1)}{k}\left(1 \pm \sqrt{\lambda/d}\right)$, and that their degrees are almost uniformly spread across color classes. 
\begin{definition}\label{def:degrees}
     For every $i,j \in [k]$ with $i\neq j$, let 
    $$C_{i,j}^+ = \left \{ v \in C_i : |\nei_{G'}(v) \cap C_j| \geq \frac{d}{k} \left(1+ \sqrt{\frac{\lambda}{d}}\right)\right \} \quad \text{and} \quad C_{i,j}^- = \left \{ v \in C_i : |\nei_{G'}(v) \cap C_j| \leq \frac{d}{k} \left(1- \sqrt{\frac{\lambda}{d}}\right)\right \}.$$ 
\end{definition}
\begin{lemma}\label{lemma:degrees}
    For every $i\neq j$, $$|C_{i,j}^+|, |C_{i,j}^-| =  O\left( \frac{\lambda n}{d}\right).$$
\end{lemma}
 \Cref{{lemma:degrees}} is proved in \cite{DF16} for $k = 3$ (see the proof of Lemma C.4 in \cite{DF16}). We include the proof of the general version for completeness. 
The proof uses the expander mixing lemma, stated below. 
    \begin{lemma}[Expander Mixing Lemma]\label{lemma:expandermixing}
Let $G=(V,E)$ be a $d$-regular $\lambda$-expander on  $n$ vertices. 
Then for all $S,T \subseteq V$, 
\[
\left| |E_G(S,T) |- \frac{d}{n}\,|S|\,|T| \right|
\le \lambda \,\sqrt{|S|\,|T|}.
\]
\end{lemma}
\begin{proof}[Proof of \Cref{lemma:degrees} ]
   For every $i, j$, by \Cref{lemma:expandermixing} applied with $S = C_{i,j}^+$ and  $T = C_j$, 
   \begin{align*}
       \left(1 + \sqrt{\frac{\lambda}{d}}\right)\frac{d}{k} |C_{i,j}^+| & \leq \left|E_{G'}\left(C_{i,j}^+, C_j \right) \right| && \text{by \Cref{def:degrees}} \\
       & \leq  \left|E_{G}\left(C_{i,j}^+, C_j \right) \right|  && \text{since $G' \subseteq G$}\\
       &\leq  \frac{d}{n}|C_{i,j}^+||C_j| + \lambda \sqrt{|C_{i,j}^+||C_j|} && \text{by \Cref{lemma:expandermixing}}. 
   \end{align*}
   Rearranging and using that $|C_j| = n/k$ gives 
 $$ |C_{i,j}^+| \leq \frac{\lambda n k  }{d} = O\left( \frac{\lambda n }{d} \right) $$
A similar argument gives 
  $$ |C_{i,j}^-| \leq \frac{\lambda n k  }{d} = O\left( \frac{\lambda n }{d} \right) . $$
\end{proof}

\subsection{Proof of \Cref{lemma:C4}}\label{sec:C4}
In this section, we prove \Cref{lemma:C4}, which is restated below. 
\spectralprops*

\cite{DF16} proved an analogous guarantee for the \emph{unnormalized} adjacency matrix $A$. We will use the following two results from \cite{DF16}. 

\begin{lemma}[Lemma C.4, 1.(c) in \cite{DF16}]\label{lemma:small_eval} For all $2 \leq i\leq n-k+1,$
    $$|\lambda_i(A)| \leq 2\lambda + O\left(\sqrt{d\lambda}\right).$$ 
\end{lemma}
\begin{remark}
Lemma C.4 in \cite{DF16} is stated for $k=3$. However, as noted in  \cite{DF16}, their proof gives the same result for general $k = O(1)$. 
\end{remark}
\begin{lemma}[Lemma B.3 in \cite{DF16}]\label{lemma:B3}
    Let $X$ be a symmetric matrix of order $n$ with eigenvalues $\lambda_1, \dots \lambda_n$ and associated eigenvectors $e_1(X), \dots e_n(X)$. For a unit vector $\bar u$ and a scalar $\alpha$, suppose that $X \bar u = \alpha \bar u + \epsilon$ for some vector $\epsilon \in \R^n$. For $\delta > 0$, let $S_{\delta} \subseteq \{1,\dots,  n\}$ be the set of those indices $i$ for which $|\lambda_i - \alpha| \leq \delta$. Then there exists a vector $\epsilon_u \in \R^n$ such that $\bar u + \epsilon_u \in \spn\left( \{ e_i(X)\}_ {i \in S_{\delta}}\right)$ and $\|\epsilon_u\|_2 \leq \|\epsilon\|_2/\delta$.  
\end{lemma}
\begin{proof}[Proof of \Cref{lemma:C4}]
We show each of the properties separately.
\paragraph{Property \ref{item:lambda1}.} Holds by \Cref{obs:e1}. 
\paragraph{Property \ref{item:lambdamid}.} Let $v_1, \dots v_n$ denote the unit norm eigenvectors of the unnormalized adjacency matrix $A$, corresponding to eigenvalues $\lambda_1(A), \dots \lambda_{n}(A)$. By the Courant–Fischer theorem, 
\begin{align*}
   \lambda_2(\bar A) & \leq \max_{x \in\spn \{ D^{1/2}v_2, \dots,  D^{1/2}v_n\}} \frac{\langle x, \bar Ax\rangle }{\|x\|_2^2}  \\
   & = \max_{y \in\spn \{ v_2, \dots , v_n\}} \frac{\langle D^{1/2}y , \bar A D^{1/2}y \rangle }{\|D^{1/2}y\|_2^2} \\
   & = \max_{y \in\spn \{ v_2, \dots,  v_n\}} \frac{\langle y , A y \rangle }{\|D^{1/2}y\|_2^2}  && \text{since $A = D^{1/2} \bar A D^{1/2}$}\\
   & = O \left(\frac{1}{d}\right)  \max_{y \in\spn \{ v_2, \dots,  v_n\}} \frac{\langle y , A y \rangle }{\|y\|_2^2}  && \text{by the model assumption} \\
   & = O \left(\frac{1}{d}\right) \lambda_2(A) \\
   & \leq  O\left(\sqrt{\frac{\lambda}{d}}\right) && \text{by \Cref{lemma:small_eval}}. 
\end{align*}
Similarly, by the Courant–Fischer theorem,   
\begin{align*}
    \lambda_{n-k+1}(\bar A) &\geq  \min_{x \in\spn \{ D^{1/2}v_1, \dots,  D^{1/2}v_{n-k+1}\}} \frac{\langle x, \bar Ax\rangle }{\|x\|_2^2}  \\
     & = \min_{y \in\spn \{ v_1, \dots , v_{n-k+1}\}} \frac{\langle D^{1/2}y , \bar A D^{1/2}y \rangle }{\|D^{1/2}y\|_2^2} \\
   & = \min_{y \in\spn \{ v_1, \dots,  v_{n-k+1}\}} \frac{\langle y , A y \rangle }{\|D^{1/2}y\|_2^2} && \text{since $A = D^{1/2} \bar A D^{1/2}$} \\
   & = O \left(\frac{1}{d}\right)  \min_{y \in\spn \{ v_1, \dots,  v_{n-k+1}\}} \frac{\langle y , A y \rangle }{\|y\|_2^2}  && \text{by the model assumption} \\
   & = O \left(\frac{1}{d}\right) \lambda_{n-k+1}(A) \\
   & \geq -  O\left(\sqrt{\frac{\lambda}{d}}\right) && \text{by \Cref{lemma:small_eval}}. 
\end{align*}
\paragraph{Property \ref{item:lambdatop}.} We will apply \Cref{lemma:B3} to show that $\bar A$ has at least $k-1$ eigenvectors with eigenvalues $-1/(k-1) \pm O(\sqrt{\lambda/d)}$.

\begin{claim}\label{claim:D12x_close_to_evect}
    Let $x \in \spn(\{\1_{C_i}\}_{i \in [k]})$ be a vector such that $\langle x, \1\rangle = 0$.  Then 
    $$\left \| \left(\bar A +\frac{1}{k-1 }I\right)D^{1/2}x \right \|^2_2 \leq O\left(\frac{\lambda}{d}\right) \| D^{1/2}x\|^2_2.$$ 
\end{claim}
\begin{proof}
By the claim assumptions, we can write $x = \sum_{i}\beta_i \1_{C_i}$ for some $\beta_i \in \R$ such that $\sum_{i} \beta_i= 0$. Without loss of generality assume $\sum_{i}\beta_i^2 =1$. Let $i \in [k]$ and let $v \in C_i$. We have 
\begin{align*}
    \left(\bar AD^{1/2}x + \frac{1}{k-1}D^{1/2}x\right)_v & =     \left(D^{-1/2}Ax + \frac{1}{k-1}D^{1/2}x\right)_v \\
    & = \frac{1}{\sqrt{\deg(v)}}\sum_{u \in \nei_{G'}(v)}x_u + \frac{\sqrt{\deg(v)}}{k-1}x_v \\
    & = \frac{1}{\sqrt{\deg(v)}}\left(\sum_{j \in [k]}|\nei_{G'}(v)\cap C_j| \beta_j  + \frac{\deg(v)}{k-1}\beta_i \right). 
\end{align*}
 Recall the sets $C_{i,j}^+$ and $C_{i,j}^-$ from \Cref{def:degrees}. 
If $v \notin C_{i,j}^+ \cup C_{i,j}^-$ for all $j \in [k]$, then we can upper-bound the absolute value as 
\begin{equation}
    \begin{aligned}\label{eq:casenormal}
    \left| \left(\bar AD^{1/2}x + \frac{1}{k-1}D^{1/2}x\right)_v \right| & \leq \frac{1}{\sqrt{\deg(v)}}\left(\sum_{j \neq i} \frac{d}{k}\left(\beta_j + \sqrt{\frac{\lambda}{d}}|\beta_j|  \right)+ \frac{d}{k}\left( \beta_i  + \sqrt{\frac{\lambda}{d}}|\beta_i|\right)\right) \\
    & = \frac{1}{\sqrt{\deg(v)}} \frac{d}{k}\sqrt{\frac{\lambda}{d}}\sum_{j \in [k]} |\beta_j| \\
    & = O(\sqrt{\lambda}), 
\end{aligned}
\end{equation}
where the second transition used that $\sum_{j \in [k]}\beta_j = 0$ and the last transition used the assumptions that $\deg(v) = \Omega(d)$, $\sum_{j \in [k]} \beta_j^2 =1$ and $k = O(1)$. 
If instead $v \in C_{i,j}^+$ or $v \in C_{i,j}^-$ for some $j \neq i$, then we can upper-bound the absolute value as 
\begin{equation}\label{eq:caseabnormal}
    \begin{aligned}
     \left| \left(\bar AD^{1/2}x + \frac{1}{k-1}D^{1/2}x\right)_v \right| & \leq \frac{1}{\sqrt{\deg(v)}}\left(\deg(v) + \frac{\deg(v)}{k-1} \right) \leq 2\sqrt{d}. 
\end{aligned}
\end{equation}
Hence, 
\begin{align*}
\left\|\bar AD^{1/2}x + \frac{1}{k-1}D^{1/2}x \right\|^2_2 & =  \sum_{v \in V}\left(\left(\bar AD^{1/2}x + \frac{1}{k-1}D^{1/2}x\right)_v\right)^2 \\
    & \leq 4d \left|\bigcup_{i,j} (C_{i,j}^+ \cup C_{i,j}^-)\right| +  O(n \lambda) && \text{by \eqref{eq:casenormal} and \eqref{eq:caseabnormal}} \\
    &= O(n \lambda) && \text{by \Cref{lemma:degrees}.} 
\end{align*}
The claim now follows, since  $\| D^{1/2}x\|^2_2 = \Omega(dn)$ by the model assumption (\Cref{def:model}). 
\end{proof}

We now obtain the following claim. 
\begin{claim}\label{claim:C4} For every $\delta > 0$, let $S_{\delta}\subseteq [n]$ denote the set of indices for which $|\lambda_i(\bar A) + \frac{1}{k-1}| \leq \delta$. Then for every $x \in \spn(\{\1_{C_i}\}_{i \in [k]})$ with  $\langle x, \1\rangle = 0$, there exists a vector $x^* \in \spn\left( \{e_i(\bar A)\}_{i \in S_{\delta}}\right)$ such that 
$$ \|{D^{1/2}x}- x^*\|_2\leq  \frac{O\left ( \sqrt{\lambda/d }\right)}{\delta}\|D^{1/2}x\|_2. $$
\end{claim}
\begin{proof}
  Let $\bar u = D^{1/2}x / \|D^{1/2}x\|_2$. By \Cref{claim:D12x_close_to_evect}, 
$$ \left\| \bar A\bar u + \frac{1}{k-1}\bar u \right\|^2_2 = O\left( \frac{\lambda}{d}\right).$$
The claim now follows from \Cref{lemma:B3}. 
\end{proof}

Let $c$ be a sufficiently large constant, and apply \Cref{claim:C4} with $\delta = c \sqrt{\lambda/d}$. Note that the space  $\{D^{1/2}x : x \in \spn\{\1_{C_i}\} \text{ and } \langle x, \1\rangle = 0\}$ has dimension $k-1$. 
It therefore follows from  \Cref{claim:C4} that $|S_{\delta}| \geq k-1$. In particular, $\bar A$ has at least $k-1$ eigenvalues in the interval $[-1/(k-1) - O(\sqrt{\lambda/d}), -1/(k-1) +  O(\sqrt{\lambda/d})]$. From Property \ref{item:lambda1} and Property \ref{item:lambdamid},  these must be the bottom $k-1$ eigenvalues. 

\paragraph{Property \ref{item:projproperty}.} Follows by setting $\delta = (\lambda/d)^{1/4}$ in \Cref{claim:C4}. 
\end{proof}

\subsection{Proof of \Cref{lemma:B_bound}}\label{sec:good-vertices} 
In this section, we prove \Cref{lemma:B_bound}, which is restated below. 
\Bbound*
For the purpose of this section, it will be more convenient to work with $k$-dimensional embeddings where we embed each vertex into the bottom $k-1$ eigenvectors of $\bar A$ \emph{as well as} the top eigenvector.
\begin{definition}\label{def:Ustar}
Let $\Ustar \in \mathbb{R}^{n\times k}$ denote the matrix whose first $(k-1)$ columns are identical to the columns of $U_{[k-1]}$, and the $k^{th}$ column is the top eigenvector of $\bar A$. 
For every vertex $v \in V$, let $f^*_v \in \R^k$ be the $v$-th row of $\Ustar$, i.e., $f^*_v \coloneqq \Ustar^\top \1_v $. 
Furthermore, for every $i \in [k]$, define
$$ \mu_i^* \coloneqq \frac{1}{|C_i|}\sum_{v \in C_i}f_v^* = \frac{1}{|C_i|}\Ustar^\top \1_{C_i}$$
\end{definition}

With this setup, we define an auxiliary set of bad vertices, $B^*$, 
consisting of vertices that have ``bad degree" (denoted $B_{\deg}$), and vertices whose embedding $f_v^*$ has ``bad spectral properties" (denoted $B_{\mathrm{spec}}$). 

\begin{definition}[$B^*$] \label{def:Bstar} 
Let
$$ B_{\deg} \coloneqq \left\{ v \in V : \deg(v) \notin \left[ \frac{d(k-1)}{k}\left(1 - \sqrt{\frac{\lambda}{d}}\right) ,\frac{d(k-1)}{k}\left(1 + \sqrt{\frac{\lambda}{d}}\right)\right] \right\}$$ and let
$$B_{\mathrm{spec}} \coloneqq \left\{ v \in V : \| f^*_v - \mu^*_{\chi(v)}\|^2_2 >  \frac{0.001k}{n} \right\}.$$ Let
$$ B^* \coloneqq B_{\deg} \cup B_{\mathrm{spec}} . $$
\end{definition}

The proof of \Cref{lemma:B_bound} consists of two parts: In \Cref{lemma:B_in_Bstar} we show that $B\subseteq B^*$. In \Cref{lemma:Bstar_bd}, we show that $|B^*| \leq O\left( \sqrt{\frac{\lambda }{d}}n\right)$.
\begin{lemma}\label{lemma:B_in_Bstar}
    Let $B$ be as in \Cref{def:goodbad}, and let $B^*$ be as in \Cref{def:Bstar}. Then 
    $$ B \subseteq B^*.$$
\end{lemma}

\begin{lemma}\label{lemma:Bstar_bd}
Let $B^*$ be as in \Cref{def:Bstar}. Then 
$$|B^*| \leq  O\left( \sqrt{\frac{\lambda }{d}}n\right). $$
\end{lemma}

\begin{proof}[Proof of \Cref{lemma:B_bound}]
    Immediate from \Cref{lemma:B_in_Bstar} and \Cref{lemma:Bstar_bd}. 
\end{proof}
The rest of this subsection is structured as follows: In \Cref{sec:mu-s}, we prove two useful properties of $\Ustar$ and of the $\mu_i^*$'s, which we will use for both \Cref{lemma:B_in_Bstar} and \Cref{lemma:Bstar_bd}. We prove \Cref{lemma:B_in_Bstar} in  \Cref{sec:B_in_Bstar} and \Cref{lemma:Bstar_bd} in \Cref{sec:Bstar_bd}

\subsubsection{Useful spectral properties}\label{sec:mu-s}
\begin{lemma}\label{lemma:mu_to_Dmu}
    For every $i \in [k]$, it holds that
    $$ \left|\left( D^{1/2} - \sqrt{\frac{d(k-1)}{k}}I_n\right)\frac{\1_{C_i}}{|C_i|} \right\|_2^2 \leq  O\left( \frac{\sqrt{d\lambda}}{n}\right).$$
\end{lemma}
\begin{proof}
Fix $i \in [k]$. Recall the sets $C_{i,j}^-$ and $C_{i,j}^+$ from \Cref{def:degrees}. Let
$C_i^+ \coloneqq \{v \in C_i: \deg(v) \geq \frac{d(k-1)}{k}(1+\sqrt{\lambda/d})\}$ and let $C_i^- \coloneqq \{v \in C_i : \deg(v) \leq \frac{d(k-1)}{k}(1-\sqrt{\lambda/d})\}$. Then $C_i^+ \subseteq \cup_j C_{i,j}^+$ and   $C_i^- \subseteq \cup_j C_{i,j}^-$, so by \Cref{lemma:degrees}, $|C_i^+|, |C_i^-| = O\left( n \lambda/d \right)$. Therefore,  
    \begin{align*}
    \left|\left( D^{1/2} - \sqrt{\frac{d(k-1)}{k}}I_n\right)\frac{\1_{C_i}}{|C_i|} \right\|_2^2  & = \frac{1}{|C_i|^2}\sum_{v \in C_i} \left(\sqrt{\deg(v)} - \sqrt{{\frac{d(k-1)}{k}}} \right)^2 \\
    & \leq \frac{1}{|C_i|^2}\sum_{v \in C_i}\left|\deg(v) - \frac{d(k-1)}{k} \right| \\
    & \leq  \frac{d}{|C_i|^2}\left( |C_i^+| +   |C_i^-| \right) + \frac{1}{|C_i|}\frac{d(k-1)}{k}\sqrt{\frac{\lambda}{d}} \\
    & =  O\left( \frac{\sqrt{d\lambda}}{n}\right). 
    \end{align*}
\end{proof}

A useful property of the $k$-dimensional color class centers $\mu_i^*$ is that they are nearly orthogonal: 
\begin{lemma}\label{lemma:mu_inner_prods}
    For every $i \in [k]$, 
    $$ \left| \| \mu_i^*\|^2_2 - \frac{1}{|C_i|}\right| \leq O\left(\sqrt{\frac{\lambda}{d}}\frac{1}{n} \right).$$
    For every $i ,j \in [k]$ with $i \neq j$, 
 $$\left| \langle \mu_i^*, \mu_j^* \rangle\right| \leq  O\left(\sqrt{\frac{\lambda}{d}}\frac{1}{n}  \right). $$
\end{lemma}
\begin{proof}
     We will show that for every $i,j \in [k]$, 
$$ \left| \langle \mu_i^*, \mu_j^*\rangle - \left\langle  \frac{1}{|C_i|} \1_{C_i},\frac{1}{|C_j} \1_{C_j} \right \rangle \right| \leq  O\left( \sqrt{\frac{\lambda}{d}}\frac{1}{n} \right). $$
The lemma then follows immediately. 

Let $i,j \in [k]$. By definition of $\mu_i^*, \mu_j^*$ (\Cref{def:Ustar}),   
\begin{align*}
    \langle \mu_i^*, \mu_j^*\rangle - \left\langle  \frac{1}{|C_i|} \1_{C_i},\frac{1}{|C_j|} \1_{C_j} \right \rangle & =  \frac{1}{|C_i||C_j|} \1_{C_i}^\top \Ustar \Ustar^\top \1_{C_j} - \frac{1}{|C_i||C_j|}   \1_{C_i}^\top \1_{C_j} \\
    & = -  \frac{1}{|C_i||C_j|}  \1_{C_i}^\top\left( I - \Ustar \Ustar^\top \right)^\top \left( I - \Ustar \Ustar^\top \right) \1_{C_j}. 
\end{align*}
where the last equality uses that $\left( I - \Ustar \Ustar^\top \right) =   \left( I - \Ustar \Ustar^\top \right)^\top \left( I - \Ustar \Ustar^\top \right) $. 
Therefore, by the Cauchy-Schwarz inequality, 
$$  \left| \langle \mu_i^*, \mu_j^*\rangle - \left\langle  \frac{1}{|C_i|} \1_{C_i},\frac{1}{|C_j|} \1_{C_j} \right \rangle \right|  \leq \left\|\left( I - \Ustar \Ustar^\top \right)\frac{1}{|C_i|}\1_{C_i} \right\|_2 \left\|\left( I - \Ustar \Ustar^\top \right)\frac{1}{|C_j|}\1_{C_j} \right\|_2  $$
The proof is completed by showing the following claim: 
\begin{claim}\label{lemma:projectnorm}
    For every $i \in [k]$, 
    $$\left\|\left( I_n - \Ustar \Ustar^\top \right)\frac{1}{|C_i|}\1_{C_i} \right\|^2_2\leq O\left( \sqrt{\frac{\lambda}{d}}\frac{1}{n}\right)$$
\end{claim}
\begin{proof}
Fix $i \in [k]$. Let $x_1 \coloneqq \1/\sqrt{n}$, and pick vectors $x_2, \dots x_k \in \R^n$ such that $x_1, x_2,\dots x_k$ is an orthonormal basis for $\spn\{\1_{C_j}\}_{j \in [k]}$. 
Expressing  $\frac{1}{|C_i|}\1_{C_i} $ in terms of this basis gives
$$  \frac{1}{|C_i|}\1_{C_i} = \sum_{j = 1}^k \beta_j x_j  $$ for some  $\beta_j \in \R$ such that $\sum_j \beta_j^2 = \|\1_{C_i}/ |C_i|\|^2_2 = k/n $. 
Since $\langle x_j, \1\rangle = 0$ for $j=2, \dots k$, we can apply \Cref{lemma:C4} to obtain vectors $x_2^*, \dots x_k^*\in \spn\{e_n, \dots e_{n-k+2}\}$ such that $\|D^{1/2}x_j - x^*_j\|^2_2 = O(\sqrt{\lambda/d})\|D^{1/2}x_j\|_2^2 = O(\sqrt{\lambda d})$ for all $j\in [2,k]$. 
Additionally, let $x_1^* =D^{1/2}x_1$.  Then $x_1^* \in \spn(e_1)$ and $\|D^{1/2}x_1 - x_1^*\|_2 = 0$. Let 
$$\delta \coloneqq \sum_{j=1}^k\beta_j (D^{1/2}x_j - x_j^*). $$
Then, using the assumption that $k = O(1)$, we get 
$$\|\delta\|^2_2 = O\left(\sum_{j}\beta_j^2\|D^{1/2}x_j - x^*_j\|^2_2  \right) =  O\left(\frac{\sqrt{ \lambda d}}{n}\right)$$
and
$$  \frac{1}{|C_i|}D^{1/2}\1_{C_i} = \delta + \sum_{j = 1}^k \beta_j x_j^*. $$

\noindent 
In particular, $ \frac{1}{|C_i|}D^{1/2}\1_{C_i} - \delta \in \spn\{e_1, e_n, e_{n-1}, \dots e_{n-k+2}\}$. 
This can equivalently be written as
\begin{equation*}\label{eq:C_i_expression2}
\frac{1}{|C_i|}D^{1/2}\1_{C_i} - \delta =\Ustar \Ustar^\top \left(\frac{1}{|C_i|}D^{1/2}\1_{C_i} - \delta\right).
\end{equation*}
Rearranging gives 
$$\left( I_n - \Ustar \Ustar^\top\right)\frac{1}{|C_i|}D^{1/2}\1_{C_i} = \left( I_n - \Ustar \Ustar^\top\right)\delta, $$ so 
\begin{equation}\label{eq:norm_term1}
     \left\|\left( I_n - \Ustar \Ustar^\top\right)\frac{1}{|C_i|}D^{1/2}\1_{C_i} \right\|^2_2=   \left\|\left( I_n - \Ustar \Ustar^\top\right)\delta \right\|^2_2 \leq \| \delta\|^2_2 \leq O\left( \frac{\sqrt{ \lambda d}}{n}\right). 
\end{equation}
 Finally, by \Cref{lemma:mu_to_Dmu}, 
\begin{align*}
  \sqrt{\frac{d(k-1)}{k}} \left \| \left( I_n - \Ustar \Ustar^\top\right)\frac{1}{|C_i|}\1_{C_i} \right\|_2 & \leq   \left\| \left( D^{1/2} - \sqrt{\frac{d(k-1)}{k}}I_n\right)\frac{\1_{C_i}}{|C_i|} \right\|_2  +  \left\|\left( I_n - \Ustar \Ustar^\top\right)\frac{1}{|C_i|}D^{1/2}\1_{C_i} \right\|_2 \\
  & \leq O\left( \frac{\sqrt[4]{\lambda d}}{{\sqrt n}}\right) 
\end{align*}
Here the first transition uses the triangle inequality and $ I_n - \Ustar \Ustar^\top \preceq I_n $, and the second transition used \Cref{lemma:mu_to_Dmu} and \Cref{eq:norm_term1}. Rearranging gives the claim. 
\end{proof}
\end{proof}

\subsubsection{Proof of \Cref{lemma:B_in_Bstar}}\label{sec:B_in_Bstar}
 \begin{proof}[Proof of \Cref{lemma:B_in_Bstar}]
     We will show that $V \setminus B^* \subseteq V \setminus B$. Suppose that $v \in V \setminus B^*$. Then $v \notin B_{\deg}$, so Property \ref{prop_goodbad_degree} in \Cref{def:goodbad} holds. We now verify that the last two properties in \Cref{def:goodbad} hold.
     \paragraph{Property \ref{prop_goodbad_own}.} We have 
     \begin{align*}
         \|f_v - \mu_{\chi(v)}\|^2_2 & \leq    \|f^*_v - \mu^*_{\chi(v)}\|^2_2 && \text{since $\U \U^{\top} \preceq \Ustar \Ustar^\top$} \\
         & \leq  \frac{0.001k}{n} && \text{by \Cref{def:Bstar}} \\
         & < \frac{0.05k}{n}. 
     \end{align*}
    \paragraph{Property \ref{prop_goodbad_other}.} Let $i$ be the index such that $v \in C_i$, and let  $j \neq i$. By the triangle inequality, 
    \begin{equation}\label{eq:fstar_bound}
    \begin{aligned}
         \left \| f^*_v - \mu^*_j\right \|_2 & \geq \left \|\mu^*_j - \mu^*_i \right \|_2 - \left \| f^*_v - \mu^*_i \right \|_2 \\
         & = \sqrt{\|\mu^*_j\|^2_2 + \|\mu^*_i\|^2_2 - 2 \langle \mu^*_i , \mu^*_j\rangle } - \left  \| f^*_v - \mu^*_i \right \|_2 \\
         & \geq \sqrt{\frac{1} {|C_i|} + \frac{1}{|C_j|} - O\left(\sqrt{ \frac{\lambda}{d}}\frac{1}{n}\right)}- \sqrt{\frac{0.001k}{n}}  && \text{by \Cref{lemma:mu_inner_prods} and \Cref{def:Bstar}} \\
         & \geq \sqrt{\frac{2k}{n}} - O\left(\left(\frac{\lambda}{d}\right)^{1/4}\frac{1}{\sqrt n} \right)-  \sqrt{\frac{0.001k}{n}}  \\
         & \geq \sqrt{\frac{k}{n}}\left( 1.38 - O\left(\left(\frac{\lambda}{d}\right)^{1/4}\right)\right). 
    \end{aligned}
    \end{equation}
To continue, we bound the difference $ \| f^*_v - \mu^*_j\|^2_2 -  \|f_v  - \mu_j\|^2_2  $
\begin{claim}\label{claim:remove_e1}

$$\| f^*_v - \mu^*_j\|^2_2 - \|f_v  - \mu_j\|^2_2  \leq O\left( \frac{1}{n}\sqrt{\frac{\lambda}{d}}\right).$$
\end{claim}
\begin{proof}
From the definition of $f_v$ and $f_v^*$ (\Cref{def:spec-emb} and \Cref{def:Ustar}), we have 
\begin{equation}\label{eq:diff}
    \| f^*_v - \mu^*_j\|^2_2 - \|f_v  - \mu_j\|^2_2 = \left \langle \1_v - \frac{1}{|C_j|}\1_{C_j}, e_1 \right \rangle^2.
\end{equation}
By \Cref{cor:m},  
\begin{equation}\label{eq:v_and_D}
    \left \langle \1_v - \frac{1}{|C_j|}\1_{C_j}, e_1 \right \rangle^2  = \frac{1}{2m} \left \langle \1_v - \frac{1}{|C_j|}\1_{C_j}, D^{1/2}\1\right \rangle^2 = O\left( \frac 1{dn} \right)\left \langle \1_v - \frac{1}{|C_j|}\1_{C_j}, D^{1/2}\1\right \rangle^2. 
\end{equation}
By the assumption that $v \notin B^*$, 
\begin{equation}\label{eq:v_term}
\begin{aligned}
   \left| \langle \1_v, D^{1/2}\1\rangle - \sqrt{\frac{d(k-1)}{k}} \right| = \left| \sqrt{\deg(v)} -  \sqrt{\frac{d(k-1)}{k}} \right| \leq O \left( (\lambda d)^{1/4}\right).
\end{aligned}
\end{equation}
where the last transition used that $v \in V \setminus B_{\deg}$. Furthermore, 
\begin{equation}
\begin{aligned}\label{D_term}
    \left| \left \langle \frac{1}{|C_j|}\1_{C_j}, D^{1/2}\1\right \rangle - \sqrt{\frac{d(k-1)}{k}} \right| & = \left| \left \langle \frac{1}{|C_j|}\1_{C_j}, D^{1/2}\1_{C_j}\right  \rangle -  \sqrt{\frac{d(k-1)}{k}} \right|\\
    &= \left| \left \langle \frac{1}{|C_j|}\1_{C_j},\left( D^{1/2} - \sqrt{\frac{d(k-1)}{k}}I \right)\1_{C_j}\right  \rangle  \right|\\
    & \leq \left \| \1_{C_j}\right\|_2\left\| \left( D^{1/2} - \sqrt{\frac{d(k-1)}{k}}I \right)\frac{\1_{C_j}}{|C_j|} \right\|_2  \\
    & \leq O\left((\lambda d)^{1/4}\right) && \text{by \Cref{lemma:mu_to_Dmu}.}
\end{aligned}
\end{equation}
Combining \Cref{eq:diff}, \Cref{eq:v_and_D}, \Cref{eq:v_term} and  \Cref{D_term} gives the claim. 
\end{proof}
We now have 
\begin{align*}
    \|f_v  - \mu_j\|^2_2 & =  \| f^*_v - \mu^*_j\|^2_2 -\left(   \| f^*_v - \mu^*_j\|^2_2- \|f_v  - \mu_j\|^2_2  \right) \\
    & \geq\frac{k}{n}\left(  1.38^2 - O\left( \left(\frac{\lambda}{d}\right)^{1/4}\right) \right)  && \text{by \Cref{eq:fstar_bound} and \Cref{claim:remove_e1}}\\
    &  \geq \frac{k}{n} && \text{for $\lambda/d$ sufficiently small.}    \\
\end{align*}

\end{proof}
 
\subsubsection{Proof of \Cref{lemma:Bstar_bd}}\label{sec:Bstar_bd}
   In this section, we bound the size of $|B^*|$.  We start by bounding the size of $B_{\deg}$. 

\begin{lemma}\label{lemma:b_deg}
    $$\left|B_{\deg}\right| \leq O\left( \frac{\lambda n}{d}\right).$$
\end{lemma}
\begin{proof}
    Let $v \in C_i$, and suppose $\deg(v) < \frac{d(k-1)}{k} \left( 1 - \sqrt{\lambda/d}\right)$. Then there exists $j \neq i $ such that $|\nei_{G'}(v)\cap C_j| <d/k\left( 1 - \sqrt{\lambda/d}\right),$ and hence $v$ belongs to the set $C_{i,j}^-$ (see \Cref{def:degrees}). 
    Similarly, if  $\deg(v) > \frac{d(k-1)}{k} \left( 1 + \sqrt{\lambda/d}\right)$, then $v \in C_{i,j}^+$  for some $j$. 
    
    Hence, 
    $ B_{\deg} \subseteq \bigcup_{i,j}\left(C_{i,j}^+  \cup C_{i,j}^-\right)$. By \Cref{lemma:degrees}, this implies 
    $$ \left| B_{\deg} \right| \leq O\left( \frac{\lambda n}{d}\right).$$
\end{proof}
\noindent 
Next, we bound the size of $B_{\mathrm{spec}}$. 
\begin{lemma}\label{lemma:Bspect}
     $$\left|B_{\mathrm{spec}}\right| \leq O\left( \sqrt{\frac{\lambda}{d}}n\right).$$
\end{lemma}
\begin{proof}
We start by bounding the sum $\sum_{v \in V}\| f^*_v - \mu^*_{\chi(v)}\|^2_2$.  
\begin{claim}\label{claim:f_v-mu}
$$ \sum_{v \in V} \|f^*_v - \mu^*_{\chi(v)}\|^2_2 = \sum_{i =1}^k |C_i| \left(\frac{1}{|C_i|} - \|\mu^*_i\|^2_2 \right)$$
\end{claim}
\begin{proof}
Let $J \coloneqq \{1\} \cup \{n-k+2, \dots n\}$ denote the indices of the eigenvectors forming the columns of  $\Ustar$. Then for every $z \in \R^n$, 
$$ \| \Ustar^\top z\|^2_2 = \sum_{j \in J}e_j^\top z z^\top e_j.$$
For every vertex $v \in V$, let $z_v \coloneqq \1_v - \frac{1}{|C_{\chi(v)}|}\1_{C_\chi(v)} \in \R^n$. Then, 
\begin{equation}\label{eq:f-mu}
     \begin{aligned}
         \sum_{v \in V} \|f^*_v - \mu^*_{\chi(v)}\|^2_2 &= \sum_{v \in  V} \left\| \Ustar^\top \left(\1_v - \frac{1}{|C_{\chi(v)}|}\1_{C_{\chi(v)}}\right)\right\|^2_2 \\
&= \sum_{v \in  V} \left\| \Ustar^\top z_v\right\|^2_2         \\  
         & = \sum_{v \in V}         \sum_{j \in J} e_j^\top z_vz_v ^\top e_j   \\
         & = \sum_{j \in J}e_j^\top\left(\sum_{v \in V}z_v z_v^\top\right)e_j. 
    \end{aligned}
\end{equation}
\noindent
For every $v \in V$, it holds that
\begin{equation}\label{eq:zvzv}
\begin{aligned}
 z_vz_v^\top & = \left( \1_v - \frac{1}{|C_{\chi(v)}|}\1_{C_{\chi(v)}}\right) \left( \1_v - \frac{1}{|C_{\chi(v)}|}\1_{C_{\chi(v)}}\right)^\top \\
 & = \1_v \1_v ^\top -  \frac{1}{|C_{\chi(v)}|}\1_{C_{\chi(v)}}\1_v^\top - \frac{1}{|C_{\chi(v)}|}\1_v \1_{C_{\chi(v)}}^\top +  \frac{1}{|C_{\chi(v)}|^2}\1_{C_{\chi(v)}} \1_{C_{\chi(v)}}^\top .
\end{aligned}
\end{equation}
Note that for every color class $C_i$, 
$$ \sum_{v \in C_i} \1_{C_i}\1_v^\top= \1_{C_i}\1_{C_i}^\top=  \sum_{v \in C_i}\1_v \1_{C_i}^\top , $$
so fixing a color class $i \in [k]$ and summing \Cref{eq:zvzv} over all $v \in C_i$ gives  
\begin{equation}\label{eq:sumzvzv}
    \begin{aligned}
    \sum_{v \in C_i}z_v z_v^\top&  = \sum_{v \in C_i}\left(\1_v \1_v ^\top -  \frac{1}{|C_i|}\1_{C_i}\1_v^\top - \frac{1}{|C_i|}\1_v \1_{C_i}^\top +  \frac{1}{|C_i|^2}\1_{C_i}  \1_{C_i}^\top\right) \\
    & = \sum_{v \in C_i}\1_v \1_v ^\top  -\frac{2}{|C_i|}\1_{C_i}  \1_{C_i}^\top + |C_i|\frac{1}{|C_i|^2}\1_{C_i}  \1_{C_i}^\top \\
    & =  \sum_{v \in C_i}\1_v \1_v ^\top - \frac{1}{|C_i|}\1_{C_i}  \1_{C_i}^\top. 
\end{aligned}
\end{equation}
Summing \Cref{eq:sumzvzv} over all $i \in [k]$ gives 
\begin{align*}
     \sum_{v \in V}z_v z_v^\top & =  \sum_{v \in V}\1_v \1_v ^\top - \sum_{i \in [k]}\frac{1}{|C_i|}\1_{C_i}  \1_{C_i}^\top = I_n - \sum_{i \in [k]}\frac{1}{|C_i|}\1_{C_i}  \1_{C_i}^\top .
\end{align*}
Substituting back into \Cref{eq:f-mu} gives 
\begin{align*}
     \sum_{v \in V} \|f^*_v - \mu^*_{\chi(v)}\|^2_2 & =  \sum_{j \in J}e_j^\top\left(\sum_{v \in V}z_v z_v^\top\right)e_j \\
     & =  \sum_{j \in J}e_j^\top\left(I_n - \sum_{i \in [k]}\frac{1}{|C_i|}\1_{C_i}  \1_{C_i}^\top\right)e_j \\
    & = \sum_{j \in J}\|e_j\|^2_2 - \sum_{j \in J}e_j^\top \left(\sum_{i \in [k]}\frac{1}{|C_i|}\1_{C_i}  \1_{C_i}^\top\right)e_j \\
    & = k - \sum_{i \in [k]}\frac{1}{|C_i|} \sum_{j \in J}e_j^\top \1_{C_i}  \1_{C_i}^\top e_j \\
    & = k -\sum_{i \in [k]}|C_i| \left\| \Ustar^\top \frac{1}{|C_i|}\1_{C_i}  \right\|^2_2 \\
    & = k  - \sum_{i \in [k]}|C_i| \|\mu_i^*\|^2_2 \\
    & = \sum_{i \in [k]}|C_i|\left( \frac{1}{|C_i|}- \|\mu_i^*\|^2_2  \right).  
\end{align*}
\end{proof}
\noindent 
By \Cref{claim:f_v-mu} and \Cref{lemma:mu_inner_prods} we now have 
\begin{equation}\label{eq:sum_bd}
   \sum_{v \in B_{\mathrm{spec}}}\|f^*_v - \mu^*_{\chi(v)}\|^2_2 \leq  \sum_{v \in V} \|f^*_v - \mu^*_{\chi(v)}\|^2_2 = \sum_{i =1}^k |C_i| \left(\frac{1}{|C_i|} - \|\mu^*_i\|^2_2 \right) \leq \sum_{i=1}^k |C_i|\cdot O\left( \sqrt{\frac{\lambda}{d}}\frac{1}{n} \right) = O\left( \sqrt{\frac{\lambda}{d}}\right).
\end{equation}
On the other hand, by \Cref{def:Bstar}, 
$$ \sum_{v \in B_{\mathrm{spec}}}\|f^*_v - \mu^*_{\chi(v)}\|^2_2  \geq |B_{\mathrm{spec}}|\cdot \frac{0.001k}{n} = |B_{\mathrm{spec}}|\cdot  \Omega\left(\frac{1}{n} \right).$$
Combining the above two equations and rearranging gives 
$ |B_{\mathrm{spec}}| = O\left( \sqrt{\frac{\lambda}{d}}n\right).$
\end{proof}
We now obtain \Cref{lemma:Bstar_bd}. 
\begin{proof}[Proof of \Cref{lemma:Bstar_bd}]
Immediate from \Cref{def:Bstar}, \Cref{lemma:b_deg} and \Cref{lemma:Bspect}.
\end{proof}

\subsection{Proof of \Cref{lemma:goodnorm}}\label{sec:goodnorm}
In this section, we prove \Cref{lemma:goodnorm}, which is restated below. 
\goodnorm*
\begin{proof}
From \Cref{def:mu} and \Cref{def:Ustar}, it holds that $\| \mu_{\chi(v)}\|_2 \leq \| \mu_{\chi(v)}^*\|_2$. Therefore, 
\begin{align*}
    \|f_v\|_2 & \leq  \| f_v-\mu_{\chi(v)} \|_2+\|\mu_{\chi(v)}\|_2  && \text{by the triangle inequality} \\
    & \leq \sqrt{\frac{0.05k}{n}}  + \sqrt{\frac{1}{|C_{\chi(v)}|}\left(1 + O\left( \sqrt{\frac{\lambda}{d}}\right)\right)} && \text{by \Cref{def:goodbad} and \Cref{lemma:mu_inner_prods}} \\
    & =  O\left(\frac{1}{\sqrt n} \right) && \text{by the assumption that $\lambda/d = O(1)$}.
\end{align*}
\end{proof}
\section{Proof of \Cref{lemma:col-count}}\label{sec:col-count}
In this section, we prove \Cref{lemma:col-count}, which is restated below for the convenience of the reader. 
\colcount*
\begin{proof}
    Let $r \geq \frac{5}{\xi^2 \cdot \zeta}(\|p\|_2^2  +\|q\|_2^2 )^{-1/2}$  be the number of samples. Let $X_1, X_2, \dots X_r \sim p $ and $Y_1, Y_2, \dots Y_r \sim q$ be the samples, and for every $i \in [r]$, $v \in [n]$, let 
    \begin{align*}
    X_{i,v} &= \mathbbm{1} \{X_i = v\} \\
    Y_{i,v} &= \mathbbm{1} \{Y_i = v\} 
    \end{align*}
  be the indicators that the $i$th samples from $p$ and $q$, respectively, are equal to $v$. Then $\mathbb{E}[X_{i,v}] = p(v)$ and $\mathbb{E}[Y_{i,v}] = q(v)$ for all $i \in [r],v \in [n]$. 
    Let 
    $$ Z = \frac{1}{r^2} \sum_{v \in [n]}w_v\sum_{i,j \in [r]} X_{i,v}Y_{j,v}.$$
    Then 
    $$ \mathbb{E}[Z] = \frac{1}{r^2} \sum_{v \in [n]} \sum_{i,j}w_vp(v)q(v) = \sum_{v \in [n]} w_vp(v)q(v) = p^\top Wq.$$
    We now compute the variance of $Z$. We have
     \begin{align*}
    Z^2 & = \left( \frac{1}{r^2}\sum_{v \in [n]}w_v\sum_{i,j \in [r]} X_{i,v}Y_{j,v} \right) \left(\frac{1}{r^2}  \sum_{u \in [n]}w_u\sum_{i',j' \in [r]} X_{i',u}Y_{j',u}\right)  \\
    & = \frac{1}{r^4}\sum_{v \in [n]} w_v^2\sum_{i,j}X_{i,v}^2Y_{j,v}^2 
    + \frac{1}{r^4}\sum_{v \in [n]}w_v^2 \sum_{i,j} \sum_{i' \neq i}X_{i,v}X_{i',v}Y_{j,v}^2  
    + \frac{1}{r^4}\sum_{v \in [n]}w_v^2 \sum_{i,j} \sum_{j' \neq j}X_{i,v}^2Y_{j,v}Y_{j',v}  \\
    & \qquad +  \frac{1}{r^4}\sum_{v \in [n]} w_v^2\sum_{i,j} \sum_{i' \neq i, j' \neq j}X_{i,v}X_{i',v}Y_{j,v} Y_{j',v}  + \frac{1}{r^4} \sum_{v} \sum_{u \neq v} w_vw_u\sum_{i,j} \sum_{i' \neq i, j' \neq j}X_{i,u}X_{i',v}Y_{j,u}Y_{j',v},
    \end{align*}
    where we are using the fact that $X_{i,u}X_{i,v} = 0$ and $Y_{i,u}Y_{i,v} = 0$ whenever $u \neq v$. Taking the expectation and using the independence between different trials, we obtain 
     \begin{align*}
     \mathbb{E}[Z^2] &= \frac{1}{r^4}r^2\sum_{v \in [n]}w_v^2 p(v)q(v)+ \frac{1}{r^4}r^2(r-1)\sum_{v \in [n]}w_v^2 p(v)^2q(v)+ \frac{1}{r^4}r^2(r-1)\sum_{v \in [n]}w_v^2p(v)q(v)^2 \\
     &  \qquad + \frac{1}{r^4}r^2(r-1)^2 \sum_{v \in [n]}w_v^2 p(v)^2q(v)^2 + \frac{1}{r^4}r^2(r-1)^2 \sum_{v \in [n]} \sum_{u \neq v}w_u w_v p(v)p(u)q(v)q(u) \\
    &  \leq \frac{1}{r^2}\sum_{v \in [n]}w_v^2  p(v) q(v) + \frac{1}{r}\sum_{v \in [n]}w_v^2 p(v)^2q(v) +  \frac{1}{r}\sum_{v \in [n]}w_v^2p(v)q(v)^2  \\
    & \qquad +\sum_{v \in [n]}w_v^2p(v)^2q(v)^2 +   \sum_{v \in [n]} \sum_{u \neq v} w_uw_vp(v)p(u)q(v)q(u) \\
    & = \frac{1}{r^2} p^\top W^2 q  + \frac{1}{r}\sum_{v \in [n]}w_v^2 p(v)^2q(v) +  \frac{1}{r}\sum_{v \in [n]}w_v^2 p(v)q(v)^2 + (p^\top W q)^2 \\
   &  \leq \frac{w_{\max}^2}{r^2} p^\top  q + \frac{w_{\max}^2}{r}\|p+q\|^3_3 + (p^\top W q)^2 \\
   & \leq  \frac{w_{\max}^2}{r^2} p^\top  q + \frac{w_{\max}^2}{r} \|p + q\|_2^{3} + (p^\top W q)^2 \\
   & \leq  \frac{w_{\max}^2}{r^2} (\|p\|_2^2  +\|q\|_2^2) + \frac{4w_{\max}^2}{r}(\|p\|_2^2 + \|q\|_2^2)^{3/2} + (p^\top W q)^2.
     \end{align*}
Thus, 
\begin{align*}
    \Var[Z] & = \mathbb{E}[Z^2] - \mathbb{E}[Z]^2  \leq  \frac{w_{\max}^2}{r^2} (\| p\|_2^2  +\| q\|_2^2) + \frac{4w_{\max}^2}{r}(\|p\|_2^2 + \|q\|_2^2)^{3/2} + (p^\top Wq)^2 - (p^\top W q)^2 \\
   &  =   \frac{w_{\max}^2}{r^2} (\|p\|_2^2  +\|q\|_2^2) + \frac{4w_{\max}^2}{r}(\|p\|_2^2 + \|q\|_2^2)^{3/2}. 
\end{align*}
By Chebyshev's inequality, 
\begin{align*}
&\Pr[|p^\top W  q - Z| \geq \xi \cdot w_{\max}(\|p\|_2^2  +\|q\|_2^2 ) ]  \\
& \leq \frac{\Var[Z]}{ \xi^2 w_{\max}^2(\|p\|_2^2  +\|q\|_2^2 )^2}  \\
& \leq \frac{ r^{-2} w_{\max}^2 (\|p\|_2^2  +\|q\|_2^2) +4r^{-1}w_{\max}^2 (\|p\|_2^2 + \|q\|_2^2)^{3/2} }{\xi^2w_{\max}^2  (\|p\|_2^2  +\|q\|_2^2 )^2} \\
& = \frac{1}{\xi^2}\left(r^{-2} (\|p\|_2^2  +\|q\|_2^2 )^{-1} + 4r^{-1}  (\|p\|_2^2  +\| q\|_2^2 )^{-1/2}\right) \\
& \leq \zeta \qquad \text{ for $r \geq \frac{5}{\xi^2 \cdot \zeta}(\|p\|_2^2  +\|q\|_2^2 )^{-1/2}$}. 
\end{align*}
\end{proof}

\end{document}